\documentclass[11pt, reqno]{amsart}
\usepackage[english]{babel}
\usepackage{physics}
\usepackage{braket}
\usepackage{dsfont}
\usepackage[table,xcdraw]{xcolor}
\usepackage{multirow}
\usepackage{graphicx}
\usepackage{tikz}
\usepackage{amssymb}
\usepackage{fontawesome}
\usetikzlibrary{matrix}
\usepackage{stmaryrd}
\usepackage{mathtools}

\usepackage[margin=1in]{geometry}
\usepackage[pagebackref, colorlinks = true, linkcolor = blue, urlcolor  = blue, citecolor = red]{hyperref}
\usepackage{caption} 
\usepackage{amsaddr}
\usepackage[capitalize]{cleveref}

\newcommand{\jewel}{\text{\faDiamond}}
\newcommand{\cuboid}{\text{\faCube}}
\renewcommand{\epsilon}{\varepsilon}
\renewcommand{\phi}{\varphi}

\makeatletter
\newtheorem*{rep@theorem}{\rep@title}
\newcommand{\newreptheorem}[2]{
\newenvironment{rep#1}[1]{
 \def\rep@title{#2 \ref{##1}}
 \begin{rep@theorem}}
 {\end{rep@theorem}}}
\makeatother

\newtheorem{thm}{Theorem}[section]
\newtheorem*{thm*}{Theorem}
\newreptheorem{thm}{Theorem}
\newtheorem{lem}[thm]{Lemma}
\newtheorem{cor}[thm]{Corollary}
\newtheorem{ex}[thm]{Example}
\newtheorem{defi}[thm]{Definition}
\newtheorem{prop}[thm]{Proposition}
\newtheorem{remark}[thm]{Remark}

\title[Polytope compatibility]{\Large{Polytope compatibility}\\\vspace{.3cm}\small{--- from quantum measurements to magic squares ---}}

\author{Andreas Bluhm$^{1,\lowercase{a}}$, Ion Nechita$^{2,\lowercase{b}}$ \and Simon Schmidt$^{3,\lowercase{c}}$}
\email{$^a$andreas.bluhm@univ-grenoble-alpes.fr}
\address{$^1$Univ. Grenoble Alpes, CNRS, Grenoble INP, LIG, 38000 Grenoble, France}

\email{$^b$ion.nechita@univ-tlse3.fr}
\address{$^2$Laboratoire de Physique Th\'eorique, Universit\'e de Toulouse, CNRS, UPS, France}

\email{$^c$s.schmidt@rub.de}
\address{$^3$QMATH, Department of Mathematical Sciences, University of Copenhagen, Universitetsparken 5, 2100 Copenhagen, Denmark}

\address{$^3$Faculty of Computer Science, Ruhr University Bochum, Universitätsstra{\ss}e 150, 44801 Bochum, Germany}

\begin{document}

\maketitle

\date{\today}

\begin{abstract}
Several central problems in quantum information theory (such as measurement compatibility and quantum steering) can be rephrased as membership in the minimal matrix convex set corresponding to special polytopes (such as the hypercube or its dual). In this article, we generalize this idea and introduce the notion of polytope compatibility, by considering arbitrary polytopes. We find that semiclassical magic squares correspond to Birkhoff polytope compatibility. In general, we prove that polytope compatibility is in one-to-one correspondence with measurement compatibility, when the measurements have some elements in common and the post-processing of the joint measurement is restricted. Finally, we consider how much tuples of operators with appropriate joint numerical range have to be scaled in the worst case in order to become polytope compatible and give both analytical sufficient conditions and numerical ones based on linear programming.
\end{abstract}

\section{Introduction}

A polytope $\mathcal P$ containing the origin can be characterized in two different but equivalent ways: 
\begin{itemize}
    \item by its facets, as an intersection of half-spaces (the ``H'' representation):
    \begin{equation}\label{eq:polytope-facets}
        \mathcal P = \bigcap_{i=1}^f \{x \in \mathbb R^g \, : \, \langle x, h_i \rangle \leq 1 \},
    \end{equation}
    \item by its extreme points, as a convex hull (the ``V'' representation):
    \begin{equation}\label{eq:polytope-extremal-points}
        \mathcal P = \operatorname{conv}\{v_i\}_{i=1}^k.
    \end{equation}
\end{itemize}
These two different points of view are graphically represented in Figure \ref{fig:polytope}. 

\begin{figure}
    \centering
    \includegraphics{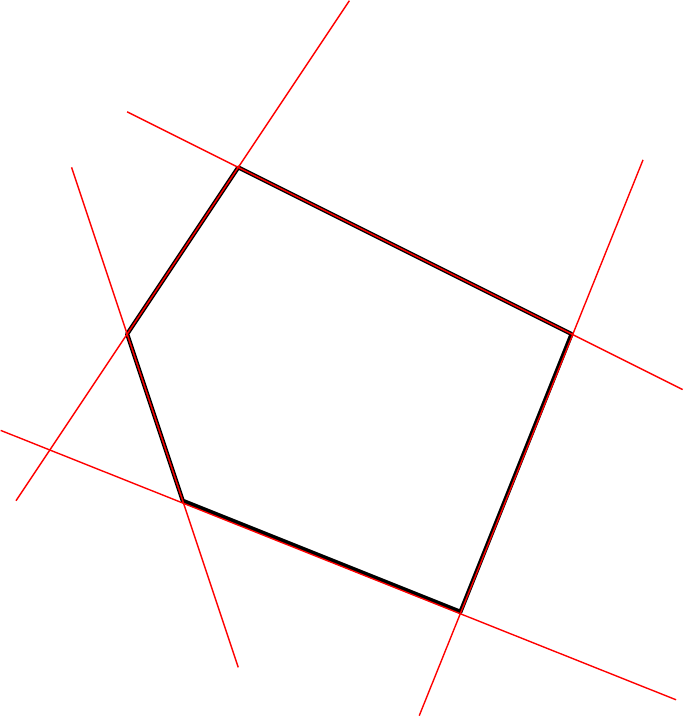} \qquad\qquad\qquad \includegraphics{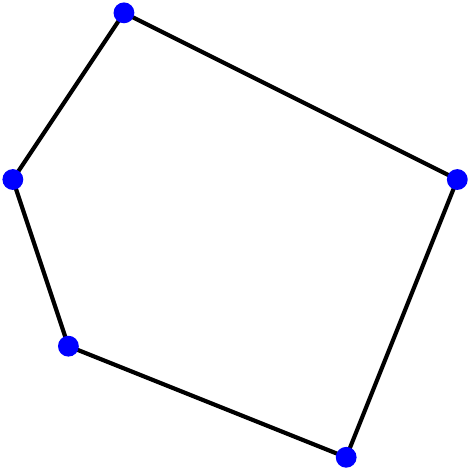}
    \caption{A polytope $\mathcal P$ described by its {facets} as an intersection of \textcolor{red}{half-spaces} (left) and as the convex hull of its \textcolor{blue}{extreme points} (right).}
    \label{fig:polytope}
\end{figure}

If we want to allow the elements of the polytope to be tuples of matrices instead of tuples of scalars, these two conditions give rise to two different and \emph{inequivalent} (in general) such matricization, which are both special cases of so-called \emph{matrix convex sets}
\begin{itemize}
    \item the facet description from Eq.~\eqref{eq:polytope-facets} generalizes to the set
    \begin{equation}
        \mathcal P_{\max}(d):= \{(A_1, \ldots, A_g) \in \mathcal M_d^{\mathrm{sa}}(\mathbb C)^g \, : \, \langle A, h_i \otimes \rho \rangle \leq 1 \quad \forall i \in [f], \, \forall \rho \in \mathcal M_d^{1,+}(\mathbb C)\},
    \end{equation}
    \item and the extreme points description from Eq.~\eqref{eq:polytope-extremal-points} generalizes to the set
    \begin{equation}
        \mathcal P_{\min}(d):= \left\{(A_1, \ldots, A_g) \in \mathcal M_d^{\mathrm{sa}}(\mathbb C)^g \, : \, \exists \,\text{POVM} \,C \, \text{s.t.} \, A_x = \sum_{i=1}^k v_i(x)C_i, \, \forall x \in [g]\right\}.
    \end{equation}    
\end{itemize}
We refer the reader to Section \ref{sec:preliminaries} for the definition of the set of density matrices $\mathcal M_d^{1,+}(\mathbb C)$, which describe quantum states, and that of a positive operator valued measure (POVM), which describe quantum measurements. As our notation suggests, $\mathcal P_{\min}$ is the smallest matrix convex set arising from $\mathcal P$ and $\mathcal P_{\max}$ is the largest.

The appearance of density matrices and POVMs in the definition of the sets $\mathcal P_{\min}$, $\mathcal P_{\max}$ suggest that there might be a link between such matrix convex sets and quantum information theory. Indeed, in the articles \cite{bluhm2018joint, bluhm2020compatibility, Bluhm2022norms}, some of the present authors realized that if one takes $\mathcal P$ to be the hypercube $[-1,1]^g$, then the following correspondence holds:
\begin{equation*}
    (2E_1-I, \ldots, 2E_g-I) \in ([-1,1]^g)_{\max} \iff \{E_i, I-E_i\}~\mathrm{POVMs}~\forall i \, .
\end{equation*}
What about $([-1,1]^g)_{\min}$? One of the defining properties that distinguish quantum mechanics from our everyday experience based on classical mechanics is the existence of \emph{incompatible measurements}, i.e., measurements that cannot be performed at the same time \cite{Heisenberg1927, Bohr1928}. Such measurements are indispensable for detecting quantum non-locality \cite{Fine1982} and can therefore be seen as a resource for many quantum information theoretic tasks similar to entanglement \cite{Brunner2014, Heinosaari2015}. For the measurements that are compatible, a joint measurement exists such that their outcomes can be recovered by classically post-processing the outcomes of the joint measurement. It turns out that membership in $([-1,1]^g)_{\min}$ is related to measurement compatibility
\begin{equation*}
    (2E_1-I, \ldots, 2E_g-I) \in ([-1,1]^g)_{\min} \iff \{E_i, I-E_i\}~\mathrm{compatible~POVMs}~\forall i \, .
\end{equation*}
The reformulation of measurement compatibility as minimal and matrix convex sets has been instrumental in finding new bounds on the maximal amount of incompatibility available in different situations. 

The success of the study of minimal and maximal matrix convex sets for the hypercube suggests the natural question: \emph{What tasks in quantum information theory can be formulated as membership in $\mathcal P_{\min}$, $\mathcal P_{\max}$ for polytopes $\mathcal P$?} This is the task this paper sets out to solve.

Motivated by the example of measurement compatibility, we define a notion of \emph{polytope operators} and \emph{polytope compatibility}. A tuple of matrices is a $\mathcal P$-operator if it is in $\mathcal P_{\max}$ and it is $\mathcal P$-compatible if it is in $\mathcal P_{\min}$. We study equivalent formulations and implications of polytope compatibility in Section \ref{sec:polytope-compatibility}. In particular, we characterize the elements which are $\mathcal P$-compatible if and only if they are $\mathcal P$-operators in Theorem \ref{thm:P-compatible-iff-P-operator} and Corollary \ref{cor:nicer-form}. An informal version is the following:
\begin{thm*}
    Let $A$ be a $g$-tuple of self-adjoint operators. Then, $A$ is $\mathcal P$-compatible for \emph{all} polytopes $\mathcal P$ such that $A$ are $\mathcal P$-operators if and only if the operators $A$ are unitarily equivalent to a direct sum of a commuting part $B$ and a non-commuting part $E$. Here, the spectrum of $B$ generates the numerical range of $A$  and the matrix convex set generated by $E$ is contained in the matrix convex set generated by the numerical range of $A$.
\end{thm*}
Then we discuss some examples before we treat the general case. In Section \ref{sec:meas-compatibility-polytope}, we review in detail which implications arise for measurement compatibility from our results on polytope compatibility. 

In Section \ref{sec:magic}, we show that another well-known problem from quantum information theory can be formulated as polytope compatibility, namely the study of \emph{quantum magic squares}. An $N \times N$ block matrix $(A_{ij})_{i,j \in [N]}$ with $d$-dimensional matrices $A_{ij}$ is a quantum magic square if both its rows $\{A_{ij}\}_{j \in [N]}$ and columns  $\{A_{ij}\}_{i \in [N]}$ form POVMs. This can be expressed with the help of the Birkhoff polytope (the set of bistochastic matrices), projected onto its supporting affine subspace. Calling this polytope $\mathcal B_N$, we arrive at the following equivalence:
\begin{equation*}
    A \in (\mathcal B_N)_{\max} \iff A \mathrm{~is~a~quantum~magic~square}.
\end{equation*}
A quantum magic square is especially simple if it has a hidden structure in terms of a tensor product of permutation matrices and a POVM. In \cite{cuevas2020quantum}, such a quantum magic square is called \emph{semiclassical}, whereas \cite{guerini2018joint} uses the term doubly normalised tensor of positive
semi-definite operators. The interest in such objects comes from the Birkhoff-von Neumann theorem, which states that the bistochastic matrices are the convex hull of the permutation matrices. The semiclassical magic squares can be seen as a matricization of this idea. It is known that not all quantum magic squares are semiclassical, but we can characterize the ones that are:
\begin{equation*}
    A \in (\mathcal B_N)_{\min} \iff A \mathrm{~is~a~semiclassical~quantum~magic~square}.
\end{equation*}
One might be tempted to conjecture that semiclassical quantum magic squares are the ones in which the row and column POVMs are compatible. However, we give an explicit example in Section \ref{sec:magic} of a quantum magic square with compatible POVMs which is \emph{not} semiclassical. Using one of our reformulations of $\mathcal P$-compatibility as factorization of an associated map through a simplex, Proposition \ref{prop:factoring}, we recover being a semiclassical quantum magic square does not only require the POVMs to be compatible, but also that the post-processing via which they arise from a joint measurement is symmetric (previously observed in \cite{guerini2018joint}).

In Section \ref{sec:POVM-common-elements} we find that polytope compatibility corresponds in general to the compatibility of POVMs with common elements under restricted post-processing. Any collection of POVMs which share elements can be represented as a hypergraph, where each POVM element is a vertex and which POVM elements belong to the same POVM is represented by hyperedges. Such hypergraphs are in one-to-one correspondence with polytopes $\mathcal P$ having vertices with rational coordinates.  Being a $\mathcal P$-operator then corresponds to being a collection of POVMs with the desired common elements. These POVMs are $\mathcal P$-compatible if and only if the POVMs are compatible and have a joint measurement from which they arise under restricted post-processing. We illustrate this in Proposition \ref{prop:3-outcomes-with-shared-effect} where we consider two POVMs with a common element $A$, such that the POVMs become
\begin{equation} \label{eq:2-3-outcome-POVMs}
    (A, B, I-A-B) \qquad (A, C, I-A-C).
\end{equation}
The polytope to which this compatibility structure corresponds is a pyramid with square basis. $(A,B,C)$ is in the minimal matrix convex set corresponding to the pyramid if and only if the two POVMs above are compatible and have a joint POVM $Q$ with five elements from which they arise as 
\begin{align*}
    A &= Q_1 \\
    B &= Q_4 + Q_5 \\
    C &= Q_3 + Q_5.
\end{align*}
We conclude the section with an explicit counterexample that not all compatible POVMs as in Eq.~\eqref{eq:2-3-outcome-POVMs} have a joint POVM of the form above, which shows that the restricted post-processing is indeed necessary.

Finally, we consider in Section \ref{sec:inclusion-constants} the question of how much  we need to shrink level-$d$ of the maximal matrix convex set for a polytope to fit it into the minimal matrix convex set for the same polytope. This leads us to the set of \emph{inclusion constants}:
\begin{equation*}
    \Delta_{\mathcal P}(d):=\Set{s \in \mathbb R^g:  s \cdot \mathcal P_{\max}(d) \subseteq  \mathcal P_{\min}(d)}.
\end{equation*}
It has been shown in \cite{bluhm2018joint} that 
\begin{equation*}
       s \in \Delta_{\mathcal [-1,1]^g}(d) \cap [0,1]^g \iff \{\widetilde E_i, I-\widetilde E_i\}~\mathrm{compatible}~\forall d-\text{dimensional~POVMs~} \{E_i, I- E_i\}\, .
\end{equation*}
Here, the $\widetilde E_i$ are the noisy versions of the $E_i$ under added white noise according to $s$, i.e.,
\begin{equation*}
   \widetilde E_i := s_i E_i + (1-s_i) \frac{I_d}{2}.
\end{equation*}
The smaller $s_i$ is, the noisier the measurement $\{\widetilde E_i, I-\widetilde E_i\}$. Thus, the set $\Delta_{\mathcal [-1,1]^g}(d)\cap [0,1]^g$ is quantifying the maximal incompatibility in $d$-dimensional dichotomic measurements.  

We give several sufficient conditions in this article for $s$ to be in $\Delta_{\mathcal P}(d)$. In Section \ref{sec:different-polytopes}, we show how to obtain inclusion constant from comparing the polytope of interest to another polytope for which the set of inclusion constants is known. In Section \ref{sec:symmetrization}, we prove sufficient conditions based on the symmetrization of the polytope inspired by previous work \cite{bluhm2020compatibility}. In particular, Proposition \ref{prop:symmetric-inclusion} generalizes findings from \cite{bluhm2018joint, helton2019dilations} and could be of independent interest. In particular, the proposition implies for polytopes such that $\mathcal P = -\mathcal P$
\begin{equation*}
    \frac{1}{2d-1} \in  \Delta_{\mathcal P}(d) \text{~for~} d \text{~even~}, \qquad   \frac{1}{2d+1} \in  \Delta_{\mathcal P}(d) \text{~for~} d \text{~odd} \, . 
\end{equation*}
We put forward a linear program in Section \ref{sec:linear-programming} and show that we can use it to compute vectors $s \in  \Delta_{\mathcal P}(d)$ for \emph{all} $d$ efficiently. Informally, the linear program is the following feasibility problem (see Theorem \ref{thm:bound-from-LP}):
\begin{thm*}
    Given $s \in \mathbb R^g$, if there exists an entrywise non-negative matrix $T$ such that 
    $$\operatorname{diag}(s_1, s_2, \ldots, s_g, 1)  =\hat V T \hat H ,$$
    then $s \in \Delta_{\mathcal P}(d)$, for all $d \geq 1$. Here, $\hat V$ is a matrix containing the vertices and $\hat H$ a matrix containing the facets of the polytope $\mathcal P$. 
\end{thm*}
We showcase the usefulness of our methods by applying them to the Birkhoff polytope and the pyramid corresponding to POVMs with shared element as in Eq.~\eqref{eq:2-3-outcome-POVMs}.

Thus, in summary, this article generalized the correspondence between measurement incompatibility and the minimal and maximal matrix convex sets of the hypercube. We find that polytope compatibility is in one-to-one correspondence with measurement compatibility with common elements and restricted post-processing. As an example, we find that being a semiclassical magic square corresponds to being Birkhoff-polytope compatible.

\section{Preliminaries}\label{sec:preliminaries}

\subsection{Notation}
For simplicity, we use the notation $[n]:= \{1, 2, \ldots, n\}$ for $n \in \mathbb N$. By $\mathcal S_n$, we denote the group of permutations on $n$ elements. For a convex set $\mathcal P \subset \mathbb R^n$, we will write $\mathcal P^\circ := \{h \in \mathbb R^n: \langle h, x \rangle \leq 1~\forall x \in \mathcal P\}$ for its polar dual.

For complex $n \times n$ matrices, we write $\mathcal M_n(\mathbb C)$. If we restrict to Hermitian matrices, we will write $\mathcal M^{\mathrm{sa}}_n(\mathbb C)$. To indicate that such a matrix $A$ is positive semidefinite, we will write $A \geq 0$, and we will denote its trace by $\mathrm{Tr}[A]$. If we restrict to density matrices, we will write  $\mathcal M_n^{1,+}(\mathbb C):= \{\rho \in \mathcal M^{\mathrm{sa}}_n(\mathbb C): \rho \geq 0,~\mathrm{Tr}[\rho]=1\}$. We write $I_n$ for the identity matrix, but we will sometimes drop the subscript. We will write $\mathcal B(\mathcal H)$ for the bounded operators on a Hilbert space $\mathcal H$.

Given a $g$-tuple $A = (A_1, \ldots, A_g)$ of self-adjoint operators in $\mathcal M_d^{\mathrm{sa}}(\mathbb C)$, their \emph{joint numerical range} is defined by
\begin{equation}
    \mathcal W(A) := \{(\langle x, A_1 x \rangle, \ldots, \langle x, A_g x \rangle) \, : \, x \in \mathbb C^d, \, \|x\|_2=1\} \subseteq \mathbb R^g,
\end{equation}
where $\|\cdot\|_2$ is the Euclidean norm. For $g=1$, one recovers the usual numerical range of a matrix, which is a convex set by the celebrated Toeplitz–Hausdorff theorem; however, for $g \geq 4$, the joint numerical range is, in general, not convex \cite{li2000convexity,gutkin2004convexity}. The convex hull of the joint numerical range is obtained by going from unit vectors (i.e.~pure quantum states) to density matrices (i.e.~mixed quantum states): 
$$ \operatorname{conv} \mathcal W(A) = \{( \operatorname{Tr}[A_1\rho], \ldots \operatorname{Tr}[A_g\rho]) \, : \, \rho \in \mathcal M_d^{1,+}(\mathbb C)\}.$$

 \subsection{Convexity}\label{sec:convexity}
 A well known way to combine two polytopes $\mathcal P_1 \subseteq \mathbb R^{k_1}$, $\mathcal P_2 \subseteq \mathbb R^{k_2}$ into another one is taking their Cartesian product:
 \begin{equation*}
\mathcal P_1 \times \mathcal P_2 := \Set{(x,y) \in \mathbb R^{k_1 + k_2}: x \in \mathcal P_1, y \in \mathcal P_2}.
\end{equation*} 

 Another one is the direct sum:
\begin{equation*}
\mathcal P_1 \oplus \mathcal P_2 := \mathrm{conv}\left(\Set{(x,0) \in \mathbb R^{k_1 + k_2}: x \in \mathcal P_1} \cup \Set{(0,y) \in \mathbb R^{k_1 + k_2}: y \in \mathcal P_2}\right),
\end{equation*} 
which is again a polytope. If both polytopes contain $0$, it holds that \cite[Lemma 2.4]{Bremner1997}
\begin{equation*}
\mathcal P_1 \oplus \mathcal P_2 = (\mathcal P_1^\circ \times \mathcal P_2^\circ)^\circ.
\end{equation*}
See also \cite[Lemma 3.6]{bluhm2020compatibility} for a short proof. Moreover, if both polytopes contain $0$, clearly
\begin{equation}\label{eq:direct-sum-smaller-Cartesian-product}
    \mathcal P_1 \oplus \mathcal P_2 \subseteq \mathcal P_1 \times \mathcal P_2,
\end{equation}
 as $0 \in \mathcal P_i$ for $i \in [2]$ implies $(p_1, 0)$, $(0, p_2) \in \mathcal P_1 \times \mathcal P_2$ for all $p_i \in \mathcal P_i$. 

\subsection{Maximal and minimal matrix convex sets}
We will now review some basic results about matrix convex sets, with a focus on minimal and maximal matrix convex sets. For more details, we refer the reader to \cite{davidson2016dilations}.

\begin{defi}
Let $g \in \mathbb N$. Moreover, let $\mathcal F(n) \subseteq \mathcal M_n^{\mathrm{sa}}(\mathbb C)^g$ for all $n \in \mathbb N$. Then, we call $\mathcal F = \bigsqcup_{n \in \mathbb N} \mathcal F(n)$ a \emph{free set}. Moreover, $\mathcal F$ is a \emph{matrix convex set} if it satisfies the following two properties for any $m$, $n \in \mathbb N$:
\begin{enumerate}
\item If $X = (X_1, \ldots, X_g) \in \mathcal F(m)$, $Y = (Y_1, \ldots, Y_g)\in \mathcal F(n)$, then $X \oplus Y := (X_1 \oplus Y_1, \ldots, X_g \oplus Y_g) \in \mathcal F(m+n)$.
\item If $X = (X_1, \ldots, X_g) \in \mathcal F(m)$ and $\Psi: \mathcal M_m(\mathbb C) \to \mathcal M_n(\mathbb C)$ is a unital completely positive (UCP) map, then $(\Psi(X_1), \ldots, \Psi(X_g)) \in \mathcal F(n)$.
\end{enumerate}
That is, a matrix convex set is a free set that is closed under direct sums and UCP maps.
\end{defi}
In particular, it follows from the definition that all sets $\mathcal F(n)$ are convex. A matrix convex set $\mathcal F$ is open/closed/bounded/compact if all levels $\mathcal F(n)$ have this property. 

Let $\mathcal C \subseteq \mathbb R^g$ be a closed convex set. Fixing $\mathcal F(1) = \mathcal C$, in most cases there are infinitely many matrix convex sets $\mathcal F$ with the same $\mathcal F(1)$. However, we can find a maximal and a minimal matrix convex set which have $\mathcal C$ as their first level. First, we consider the maximal matrix convex set for $\mathcal C$ \cite[Definition 4.1]{davidson2016dilations}:
\begin{align*}
\label{eq:Wmax}&\mathcal C_{\max}(n) :=\\
& \Set{X \in \mathcal M_n^{\mathrm{sa}}(\mathbb C)^g: \sum_{i = 1}^g c_i X_i \leq \alpha I_n, \quad \forall \, c \in \mathbb R^g, \forall \alpha \in \mathbb R~\mathrm{\ s.t.\ }~\mathcal C \subseteq \Set{x \in \mathbb R^g: \langle c, x\rangle \leq \alpha}}.
\end{align*}
Note that if $\mathcal C$ is a polyhedron, only finitely many hyperplanes need to be considered.

The minimal matrix convex set associated with $\mathcal C$ is defined as \cite[Eq.~(1.4)]{passer2018minimal}:
\begin{equation}\label{eq:definition-min}
    \mathcal C_{\min}(n) := \Set{X \in \mathcal M_n^{\mathrm{sa}}(\mathbb C)^g: X= \sum_{j} z_j \otimes Q_j,  z_j \in \mathcal C~ \forall j, Q_j \geq 0~\forall j, \sum_{j} Q_j =I_n}.
\end{equation}
Note that if $\mathcal C$ is a polytope, i.e.~it has finitely many extreme points, the number of terms in the decomposition above can be taken to be the number of
extreme points of $\mathcal C$. 

An equivalent definition of the minimal matrix convex set is the one used in \cite{davidson2016dilations} as
\begin{equation}\label{eq:alternative-definition-min}
    \mathcal C_{\min}(n) :=\left\{ X \in \mathcal M_n^\mathrm{sa}(\mathbb C)^g: \exists \text{~pairwise-commuting~normal~dilation~}N\text{~of~}X{~s.t.~}\sigma(N) \subseteq \mathcal C \right\},
\end{equation}
where $\sigma(N)$ is the joint spectrum of the pairwise-commuting normal dilation $N$. We recall that $N \in \mathcal B(\mathcal H)^g$ is a dilation of $X\in \mathcal M_n^\mathrm{sa}(\mathbb C)^g$ if there exists an isometry $V:\mathbb C^n \hookrightarrow \mathcal H$ such that $X_i = V^\ast N_i V$ for all $i \in [g]$.

\begin{remark}\label{rem:finite-dilations-and-POVMs-in-min}
    To go from Eq.~\eqref{eq:alternative-definition-min} to Eq.~\eqref{eq:definition-min}, we can use Naimark's dilation theorem. In order to go from Eq.~\eqref{eq:definition-min} to Eq.~\eqref{eq:alternative-definition-min}, we can use the construction used in the proof of Theorem 7.1 in \cite{davidson2016dilations}. Theorem 7.1 in \cite{davidson2016dilations} also implies that we can restrict the dilation $N$ to be a tuple of self-adjoint $nm$-dimensional matrices and the POVM in Eq.~\eqref{eq:alternative-definition-min} to have $m$ outcomes, where $m = 2n^2(g+1)+1$.
\end{remark}
 
We will also use the notion of \emph{inclusion constants}, i.e., constants for which the inclusion
\begin{equation*}
 s\cdot \mathcal C_{\max} \subseteq \mathcal C_{\min}
\end{equation*}
holds. Allowing for different scalings in each direction, the \emph{(asymmetrically) scaled} matrix convex set is 
\begin{equation*}
s \cdot \mathcal C_{\max}:= \Set{(s_1 X_1, \ldots, s_g X_g): X \in \mathcal C_{\max}}.
\end{equation*}
\begin{defi}\label{def:Delta-MCS}
Let $d$, $g \in \mathbb N$ and $\mathcal C \subset \mathbb R^g$. The \emph{inclusion set} is defined as 
\begin{equation*}
\Delta_{\mathcal C}(d):=\Set{s \in \mathbb R^g:  s \cdot \mathcal C_{\max}(d) \subseteq  \mathcal C_{\min}(d)}.
\end{equation*}
\end{defi}
Note that $\Delta_{\mathcal C}(d)$ is a convex set, because both $\mathcal C_{\min}(d)$ and $\mathcal C_{\max}(d)$ are.

Finally, we will show that if $\mathcal C$ is compact, then its corresponding minimal matrix convex set is also compact. This has already been pointed out in \cite{Passer2019} and we include the proof here for convenience. 
\begin{lem}\label{lem:min-closed}
    Let $g \in \mathbb N$ and let  $\mathcal C \subset \mathbb R^g$ be a compact convex set. Then, $\mathcal C_{\min}$ is compact (i.e., $\mathcal C_{\min}(n)$ is compact for all $n \in \mathbb N$).
\end{lem}
\begin{proof}
For $X \in \mathcal C_{\min}(n)$, Remark \ref{rem:finite-dilations-and-POVMs-in-min} implies that we can write 
    \begin{equation*}
        X= \sum_{i=1}^{2n^2(g+1)+1} z_i \otimes Q_i
    \end{equation*}
    for a POVM $Q$ and $z_i \in \mathcal C$ for all $i \in [2n^2(g+1)+1]$.
    As $\mathcal C$ and the set of POVMs with a fixed number of outcomes are both compact, $\mathcal C_{\min}(n)$ is compact as the image of a compact set under a continuous function.
\end{proof}

\subsection{Incompatible measurements in quantum mechanics}

As this will be our guiding example in this work, we will give a short introduction to measurement incompatibility in quantum mechanics. For background concerning the mathematics of quantum mechanics, see \cite{Heinosaari2011, watrous2018theory}. A quantum mechanical measurement with $k$ outcomes on a quantum system of dimension $d$ is described by a collection of positive operators $\Set{E_j}_{j \in [k]} \subset \mathcal M_d^{\mathrm{sa}}$, $E_j \geq 0$ for all $j \in [k]$, such that
\begin{equation} \label{eq:POVM-normalization}
\sum_{j = 1}^k E_j = I_d.
\end{equation}
The set $\Set{E_j}_{j \in [k]}$ is called a \emph{positive operator-valued measure (POVM)} and its elements $E_j$ are referred to as \emph{effects}. If the condition Eq.~\eqref{eq:POVM-normalization} is replaced by the sub-normalization
\begin{equation*}
    \sum_{j = 1}^k E_j \leq I_d,
\end{equation*}
then $\Set{E_j}_{j \in [k]}$ is called a \emph{sub-POVM}.

An important characteristic distinguishing quantum mechanics from classical mechanics is that measurements can be \emph{incompatible} (see \cite{Heinosaari2016} for an introduction). A collection of POVMs is compatible if they arise as marginals from a joint POVM: 
\begin{defi}[Compatible POVMs] \label{def:jointPOVM}
Let $\Set{E_j^{(i)}}_{j \in [k_i]}$ be a collection of $d$-dimensional POVMs, where $k_i \in \mathbb N$ for all $i \in [g]$, $d$, $g \in \mathbb N$. The POVMs are \emph{compatible} if there is a $d$-dimensional joint POVM $\Set{R_{j_1, \ldots, j_g}}$ with $j_i \in [k_i]$ such that for all $u \in [g]$ and $v \in [k_u]$,
\begin{equation*}
E_v^{(u)} = \sum_{\substack{j_i \in [k_i] \\ i \in [g] \setminus \Set{u}}} R_{j_1, \ldots, j_{u-1},v,j_{u+1}, \ldots j_g}.
\end{equation*}
\end{defi}

Not all measurements in quantum mechanics are compatible. For projective measurements (POVMs in which all effects are orthogonal projections), compatibility is equivalent to their effects pairwise commuting. There is an equivalent definition of joint measurability \cite[Eq.~16]{Heinosaari2016}, formulated in terms of classical post-processing, which we will also use in this work. Measurements are compatible if and only if there is a joint measurement from which their outcomes can be obtained with the help of classical randomness, i.e., \emph{classical post-processing}. 
\begin{lem} \label{lem:post-processing}
Let $E^{(i)} \in \mathcal M_d^{\mathrm{sa}}(\mathbb C)^{k_i}$, $i \in [g]$, be a collection of POVMs, where $k_i \in \mathbb N$ for all $i \in [g]$, $d$, $g \in \mathbb N$. These POVMs are compatible if and only if there is some $m \in \mathbb N$ and a POVM $M \in \mathcal M_d^{\mathrm{sa}}(\mathbb C)^{m}$ such that
\begin{equation*}
E^{(i)}_{j} = \sum_{\lambda=1}^m p_\lambda(j|i)M_\lambda
\end{equation*}
for all $j \in [k_i]$, $i \in [g]$ and some conditional probabilities $p_\lambda(j|i)$.
\end{lem}
For a collection of dichotomic measurements $\{E_i, I - E_i\}$, $i \in [g]$, we will for simplicity say that the effects $\{E_i\}_{i \in [g]}$ are compatible, since they completely determine the corresponding measurements due to normalization.

\subsection{General probabilistic theories}
In this section, we briefly introduce the formalism to describe a class of generalization of quantum mechanics, the \emph{general probabilistic theories} (GPTs). For more background, see \cite{lami2018non}.

Any GPT corresponds to a triple $(V,V^+,\mathds 1)$, where $V$ is a vector space with  a proper cone $V^+$
and $\mathds 1$
is an order unit in the dual cone $A^+=(V^+)^*\subset V^*=A$. We assume here that $V$ is finite dimensional.  The set of states of the system is identified as the  subset 
\[
K:=\{v\in V^+,\ \langle\mathds 1,v\rangle=1\}.
\]
It holds that $K$ is compact and  convex and is a base of the cone $V^+$. It is also possible to start with any compact convex set $K$ as a state space and construct the corresponding GPT $(V(K), V(K)^+, \mathds{1}_K)$ from there.  Indeed, let $K$ be such a set and let $A(K)$ be the set of affine functions $K\to \mathbb R$. Then $A(K)$ is a finite dimensional vector space and the subset of affine functions which are positive on $K$, $A(K)^+$, is a proper cone in $A(K)$. Let $\mathds 1_K$ be the constant function, then $\mathds 1_K$ is an order unit. We put $V(K)=A(K)^*$, $V(K)^+=(A(K)^+)^*$. Then $V(K)^+$ is a proper cone in $V(K)$  and $K$ is affinely isomorphic to the base of $V^+$, determined by $\mathds 1_K$. 

\begin{ex}\label{ex:CM} Any \emph{classical system} is described by the triple $\mathrm{CM}_d := (\mathbb R^d,\mathbb R^d_+,1_d)$, $d \in \mathbb N$, where $\mathbb R^d_+$
denotes the set of elements with non-negative coordinates and $1_d=(1,1,\dots,1)\in \mathbb R^d$.
Then $(\mathbb R^d)^*=\mathbb R^d$ with duality given by the standard inner product and the simplicial cone
$\mathbb R^d_+$ is self-dual.  The classical state space is the probability simplex
\begin{align*}
\Delta_d=\left\{x=(x_1,\dots,x_d)\in \mathbb R^d,\ x_i\ge 0,\ \sum_{i=1}^d x_i=1\right\}= \{x\in \mathbb R^d_+,\ \langle x, 1_d\rangle =1\}.
\end{align*}
\end{ex}

\begin{ex}\label{ex:QM} \emph{Quantum mechanics} corresponds to the triple $\mathrm{QM}_d:=(\mathcal M_d^{\mathrm{sa}}(\mathbb C), \mathrm{PSD}_d ,\operatorname{Tr})$, $d \in \mathbb N$, where  $\mathrm{PSD}_d$ is the cone of $d \times d$ positive semidefinite complex, self-adjoint matrices, and $\operatorname{Tr}$ is the usual, un-normalized, trace. As in the case of classical systems described above, the $\mathrm{PSD}_d$ cone is self-dual.  The quantum state space is the set of density matrices $\mathcal M_n^{1,+}(\mathbb C)$.
\end{ex}

A \emph{channel} between GPTs $(V_i,V_i^+,\mathds 1_i)$ with state spaces $K_i$, $i \in [2]$, is a linear map $\Phi: V_1 \to V_2$ such that $\Phi(V_1^+) \subseteq V_2^+$ (i.e., the map is a positive map between the corresponding ordered vector spaces) and such that it maps states of one GPT to states of the other: $\Phi(K_1) \subseteq K_2$. Note that $\Phi(K_1) \subseteq K_2$ implies $\Phi(V_1^+) \subseteq V_2^+$, since $K_i$ is a base of $V_i^+$ and the map $\Phi$ is linear.

\section{Polytope compatibility} \label{sec:polytope-compatibility}

  In this section, we introduce the notion of polytope compatibility, or $\mathcal P$-compatibility for a fixed polytope $\mathcal P$. 

\subsection{Equivalent characterizations}

We start with the definition of tuples of matrices being $\mathcal P$-operators and $\mathcal P$-compatible. 

\begin{defi} \label{defi:P-operators-P-compatible}
    Let  $d$, $g$, $k \in \mathbb N$ and let $\mathcal P$ be a polytope with $k$ extreme points $v_1, \ldots, v_k \in \mathbb R^g$ such that $0 \in \operatorname{int} \mathcal P$. Let $A = (A_1, \ldots, A_g) \in \mathcal M_d^{\mathrm{sa}}(\mathbb C)^g \cong \mathbb R^g \otimes \mathcal M_d^{sa}(\mathbb C)$ be a $g$-tuple of Hermitian matrices. We say that 
    \begin{itemize}
        \item $A$ are \emph{$\mathcal P$-operators} if, for any dual hyperplane $h \in \mathcal{P}^\circ$ and any density matrix $\rho \in \mathcal M_d^{1,+}(\mathbb C)$, 
        $$\langle A, h \otimes \rho \rangle \leq 1.$$
        \item $A$ are \emph{$\mathcal P$-compatible} if there exists a POVM $C = (C_1, \ldots C_k)$ in $\mathcal M_d(\mathbb C)^k$ such that
        $$A = \sum_{i=1}^k v_i \otimes C_i.$$
    \end{itemize}
\end{defi}

We can express these definitions naturally in the language of minimal and maximal matrix convex sets arising from $\mathcal P$.

\begin{prop} \label{prop:operator_compatible_max_min}
 Let  $d$, $g$, $k \in \mathbb N$  and let $\mathcal P$ be a polytope with $k$ extreme points $v_1, \ldots, v_k \in \mathbb R^g$ such that $0 \in \operatorname{int} \mathcal P$. Let $A = (A_1, \ldots, A_g) \in \mathcal M_d^{\mathrm{sa}}(\mathbb C)^g \cong \mathbb R^g \otimes \mathcal M_d^{sa}(\mathbb C)$ be a $g$-tuple of Hermitian matrices. Then, $A$ are \emph{$\mathcal P$-operators} if and only if $A \in \mathcal{P}_{\max}(d)$. Moreover,  $A$ are \emph{$\mathcal P$-compatible} if and only if $A \in \mathcal{P}_{\min}(d)$.
\end{prop}
\begin{proof}
The proof of the first assertion is straightforward. For the second assertion, we note the following: Let $(C_1^\prime, \ldots, C_n^\prime)$ be a POVM for $n \in \mathbb N$ and let $z_j \in \mathcal P$ for all $j \in [n]$. Then, every $z_j$ is a convex combination of extreme points, i.e., $z_j  = \sum_{i=1}^k w_i^{(j)}v_i$ for all $j \in [n]$, where $w_i^{(j)} \geq 0$ for all $i \in [k]$ and $\sum_{i = 1}^k w_i^{(j)} = 1$ for all $j \in [n]$. Thus, we can always write 
\begin{equation*}
    \sum_{j = 1}^n z_j \otimes C^\prime_j = \sum_{j = 1}^n \left( \sum_{i = 1}^k w_i^{(j)}v_i \right) \otimes C^\prime_j = A = \sum_{i=1}^k v_i \otimes C_i,
\end{equation*}
where $C_i = \sum_{j = 1}^n w_i^{(j)} C^\prime_j$. It can readily be verified that $(C_1, \ldots, C_k)$ is also a POVM.
\end{proof}

\begin{remark}
    In the scalar case $d=1$, the tuple $A$ being \emph{$\mathcal P$-operators} and $A$ being \emph{$\mathcal P$-compatible} coincide and the two notions are equivalent to $A \in \mathcal P$. They correspond, respectively, to the hyperplane (Eq.~\eqref{eq:polytope-facets}) and to the extreme point (Eq.~\eqref{eq:polytope-extremal-points}) definition of the polytope $\mathcal P$. They can therefore be thought of as two different ways of ``quantizing'' what it means to be a polytope. 
\end{remark}

A polytope $\mathcal P$ containing $0$ in its interior can be characterized (in the hyperplane representation) as the intersection of finitely many halfspaces
$$\mathcal P = \bigcap_{j=1}^r \{ x \in \mathbb R^g \, : \, \langle h_j, x \rangle \leq 1\},$$
where the vectors $h_1, \ldots, h_r \in \mathbb R^g$ define the facets of $\mathcal P$. In this picture, one can state the following nice description of the set $\mathcal P_{\max}$.
\begin{prop}\label{prop:P-hj}
 Let  $d$, $g$, $r \in \mathbb N$ and let $\mathcal P$ be a polytope in $\mathbb R^g$, containing $0$ in its interior, defined by the facet vectors $h_1, \ldots, h_r \in \mathbb R^g$ as above. Then, for all $d \geq 1$,
$$\mathcal P_{\max}(d) = \left\{A \in \mathcal M_d^{\mathrm{sa}}(\mathbb C)^g \, : \, \forall j \in [r], \, \sum_{x=1}^g h_j(x) A_x \leq I_d\right\}.$$
\end{prop}
\begin{proof}
    The proof follows immediately from the fact that the extreme points of the dual polytope $\mathcal P^\circ$ are a subset of the $h_j$'s. 
\end{proof}

There are ways to rewrite the condition that $A \in \mathcal P_{\min}(d)$ which will be useful later on.

\begin{remark} \label{rem:factorization}
    If $A \in \mathcal P_{\min}(d)$ for a polytope $\mathcal P$ with $0 \in \operatorname{int} \mathcal P$, then for all $\rho \in \mathcal M^{1,+}_d(\mathbb C)$
    $$(\langle A_x, \rho \rangle)_{x \in [g]}  = V \cdot (\langle C_i, \rho \rangle)_{i \in [k]} \in \mathcal P,$$
    where $(C_1, \ldots, C_k)$ is a POVM and $V \in \mathcal M_{g \times k}(\mathbb R)$ is the matrix having the extreme points of $\mathcal P$ as columns: 
    $$\forall x \in [g], \, \forall i \in [k], \qquad V_{x,i} = v_i(x).$$
\end{remark}

One can relax the POVM condition in the definition of $P$-compatibility to $C$ being a sub-POVM.
\begin{prop} \label{prop:sub-POVM-compatibility}
Let  $d$, $g$, $k \in \mathbb N$ and let $\mathcal P$ be a polytope with $k$ extreme points $v_1, \ldots, v_k \in \mathbb R^g$ such that $0 \in \operatorname{int} \mathcal P$. Let $A = (A_1, \ldots, A_g) \in \mathcal M_d^{\mathrm{sa}}(\mathbb C)^g$ be a $g$-tuple of Hermitian matrices. Then, $A$ are $\mathcal P$-compatible if and only if there exists a sub-POVM $C = (C_1, \ldots C_k)$ in $\mathcal M_d^{\mathrm{sa}}(\mathbb C)^k$
such that
$$A = \sum_{i=1}^k v_i \otimes C_i.$$
\end{prop}
\begin{proof}
    Assume $A = \sum_i v_i \otimes C_i$ for a sub-POVM $C$ and denote $C_0 := I_d - \sum_i C_i \geq 0$. Since $0 \in P$, there is a probability vector $\pi \in \mathbb R^k$ such that 
    $$\sum_{i=1}^k \pi_i v_i = 0.$$
    Write
    $$A = \sum_{i=1}^k v_i \otimes C_i + \left(\sum_{i=1}^k \pi_i v_i\right) \otimes C_0 = \sum_i v_i \otimes C_i',$$
    where $C_i':= C_i + \pi_i C_0$ for all $i \in [k]$ forms a POVM.
\end{proof}

\begin{prop} \label{prop:num_range_P_operators}
     Let  $d$, $g \in \mathbb N$. A tuple $A=(A_1,\ldots, A_g) \in \mathcal M^{\mathrm{sa}}_d(\mathbb C)^g$ consists of $\mathcal P$-operators for a polytope $\mathcal P$ with $0 \in \operatorname{int} \mathcal P$ if and only if for all density matrices $\rho \in \mathcal M_d^{1,+}(\mathbb C)$, 
        $$(\langle A_x, \rho \rangle)_{x \in [g]} \in  \mathcal P.$$
    In other words, $A$ are $\mathcal P$-operators if and only if the joint numerical range $\mathcal W(A)$ of $A$ is contained in $\mathcal P$:
    $$\mathcal W(A) \subseteq \mathcal P \iff \operatorname{conv} \mathcal W(A) = \left\{(\langle A_x, \rho \rangle)_{x \in [g]}\right\}_{\rho \in \mathcal M_d^{1,+}(\mathbb C)} \subseteq \mathcal P.$$
\end{prop}

\begin{proof}
This follows from the definition as we can write 
\begin{equation*}
    \langle A, h \otimes \rho \rangle = \sum_{i = 1}^g h_i \mathrm{Tr}[A_i \rho]
\end{equation*}
and use the bipolar theorem \cite[Theorem IV.1.2]{Barvinok2002} to conclude that $\mathcal P^{\circ \circ} = \mathcal P$.
\end{proof}

Another way to understand polytope compatibility is as factorization of positive maps through ordered vector spaces.

\begin{prop}\label{prop:factoring}
Let  $d$, $g$, $k \in \mathbb N$ and let $\mathcal P$ be a polytope with $k$ extreme points $v_1, \ldots, v_k \in \mathbb R^g$ such that $0 \in \operatorname{int} \mathcal P$. Let $A = (A_1, \ldots, A_g) \in \mathcal M_d^{\mathrm{sa}}(\mathbb C)^g$ be a $g$-tuple of Hermitian matrices. Let us consider the map
$\mathcal A: \mathcal M_d^{\mathrm{sa}}(\mathbb C) \to \mathbb R^g$,
\begin{equation*}
    \mathcal A(X) = (\mathrm{Tr}[A_1 X], \ldots, \mathrm{Tr}[A_g X]).
\end{equation*}
Then, 
\begin{enumerate}
    \item $A$ are $\mathcal P$-operators if and only if $\mathcal A$ is a channel between $(\mathcal M_d^{\mathrm{sa}}, \mathrm{PSD}_d, \mathrm{Tr})$ and $(V(\mathcal P), V(\mathcal P)^+, \mathds{1}_{\mathcal P})$.
    \item $A$ are $\mathcal P$-compatible if and only if in addition $\mathcal A$ factors through $\Delta_k$.
\end{enumerate}
\end{prop}
\begin{proof}
    For the first point, it is clear that if $A$ are $\mathcal P$-operators, then $\mathcal A$ is a channel of the required form. Conversely, if $\mathcal A$ is a channel of the form in the assertion, then in particular $\mathcal W(A) \subseteq \mathcal P$. The first assertion then follows from Proposition \ref{prop:num_range_P_operators}.

    For the second point, if the $A$ are $\mathcal P$-compatible, the map $\mathcal Q: \mathcal M_d^{\mathrm{sa}} \to \mathbb R^k$,
    \begin{equation} \label{eq:Q_map}
    \mathcal Q(X) = (\mathrm{Tr}[Q_1 X], \ldots, \mathrm{Tr}[Q_k X]),
\end{equation}
sends $\mathcal M_d^{1,+}(\mathbb C)$ to $\Delta_k$. The matrix $V$ as defined in Remark \ref{rem:factorization} then maps $\Delta_k$ into $\mathcal P$ as required. Conversely, it can be checked that any channel mapping $\mathcal M_d^{1,+}(\mathbb C)$ to $\Delta_k$ gives rise to a POVM $(Q_1, \ldots, Q_k)$ as in Eq.~\eqref{eq:Q_map}. Any map $\nu: \Delta_k \to \mathcal P$ satisfies $\nu(\delta_i) = z_i \in \mathcal P$ for the vertices $\delta_i$ of $\Delta_k$. Thus, the factorization implies
\begin{equation*}
    (\mathrm{Tr}[A_1 \rho], \ldots, \mathrm{Tr}[A_g \rho]) = \sum_{i = 1}^k z_i \mathrm{Tr}[Q_i \rho] \qquad \forall \rho \in \mathcal M_d^{1,+}(\mathbb C),
\end{equation*}
hence $A_j = \sum_{i = 1}^k z_i(j) Q_i$. Therefore, $A \in \mathcal P_{\min}$ and the assertion follows from Proposition \ref{prop:operator_compatible_max_min}.
\end{proof}

\subsection{Tuples that are \texorpdfstring{$\mathcal P$}{}-compatible if and only if they are \texorpdfstring{$\mathcal P$}{}-operators} 
In this subsection, we give a characterization of tuples for which $\mathcal P$-compatibility is equivalent to being $\mathcal P$-operators. We first need the following lemma.

\begin{lem} \label{lem:A-in-min}
Let  $d$, $g \in \mathbb N$ and let $A \in \mathcal M_d^{\mathrm{sa}}(\mathbb C)^g$ such that $0 \in \operatorname{int} (\operatorname{conv}\mathcal W(A))$. If $A$ are $\mathcal P$-compatible for all polytopes $\mathcal P$ such that $A$ are $\mathcal P$-operators, then $A  \in (\operatorname{conv}(\mathcal W(A)))_{\min}$.
\end{lem}
\begin{proof}
    The assertion is not obvious, because $\operatorname{conv}(\mathcal W(A)$ might not be a polytope. It is easy to see that $\mathcal W(A)$ is compact as the image of the unit sphere under a continuous map, thus $\operatorname{conv}(\mathcal W(A))$ is compact as well \cite[Theorem 5.35]{aliprantis2013infinite}. Since every convex body can be approximated to arbitrary precision by a polytope, for any $\epsilon > 0$, there is a polytope $\mathcal P^\epsilon$ such that $\operatorname{conv}(\mathcal W(A)) \subseteq \mathcal P^{\epsilon}$ and such that $\max_{x \in \mathcal P^\epsilon} \min_{y \in \operatorname{conv}(\mathcal W(A))} \| x - y\|_2 \leq \epsilon$ \cite[Section 4]{bronstein2008approximation}. 
    
    Let us assume that $A \not \in (\operatorname{conv}(\mathcal W(A)))_{\min}$. Since $\operatorname{conv}(\mathcal W(A))$ is compact, Lemma \ref{lem:min-closed} implies that $(\operatorname{conv}(\mathcal W(A)))_{\min}$ is compact as well. Hence, there is an $\eta > 0$ such that
    \begin{equation}\label{eq:A-far-from-min}
        \min_{B \in (\operatorname{conv}(\mathcal W(A)))_{\min}(d)} \|A-B\| = \eta,
    \end{equation}
    where we take the norm to be $\|(X_1, \ldots, X_g)\| = \sum_{i=1}^g \|X_i\|_\infty$ for later convenience.

    Now we will argue that $A \not \in \mathcal P^{\eta/(2g)}_{\min}$. If the tuple $A$ was in $\mathcal P^{\eta/(2g)}_{\min}$, then we could write
    \begin{equation*}
        A = \sum_i p_i \otimes Q_i
    \end{equation*}
    with $p_i \in \mathcal P^{\eta/(2g)}$ and a POVM $Q$. From the definition of the approximating polytope, for any $i$, there exists a $z_i \in \operatorname{conv}(\mathcal W(A))$ such that $\|z_i - p_i\|_2 \leq \frac{\eta}{2g}$. Defining 
    \begin{equation*}
        B = \sum_i z_i \otimes Q_i,
    \end{equation*}
    it holds that $B \in \operatorname{conv}(\mathcal W(A)))_{\min}(d)$ by construction, but also $\|A-B\| \leq \eta/2$, as for all $j \in [g]$,
    \begin{equation*}
       -\frac{\eta}{2g} I \leq -\sum_{i} |p_i(j) - z_i(j)| Q_i \leq A_j-B_j \leq \sum_{i} |p_i(j) - z_i(j)| Q_i \leq \frac{\eta}{2g} I.
    \end{equation*}
    This contradicts Eq.~\eqref{eq:A-far-from-min}, hence $A \not \in \mathcal P^{\eta/(2g)}_{\min}$.

    However, $\operatorname{conv}(\mathcal W(A)) \subseteq \mathcal P^{\eta/(2g)}$ by construction, which implies that $A$ is a $\mathcal P^{\eta/(2g)}$-operator by Proposition \ref{prop:num_range_P_operators}. Using the assumption that $A$ is $\mathcal P$-compatible for all polytopes $\mathcal P$ such that $A$ are $\mathcal P$-operators, it follows that $A \in \mathcal P^{\eta/(2g)}_{\min}$, which is a contradiction. Hence, we have to conclude that $A \in (\operatorname{conv}(\mathcal W(A)))_{\min}$.
\end{proof}

Now we can give a characterization of the operators for which being $\mathcal P$-operators and $\mathcal P$-compatible is the same.
\begin{thm}\label{thm:P-compatible-iff-P-operator}
    Let  $d$, $g \in \mathbb N$ and let $A \in \mathcal M_d^{\mathrm{sa}}(\mathbb C)^g$ such that $0 \in \operatorname{int} (\operatorname{conv}\mathcal W(A))$. Then, $A$ is $\mathcal P$-compatible for \emph{all} polytopes $\mathcal P$ such that they are $\mathcal P$-operators if and only if the operators $A$ admit a pairwise commuting dilation $N$ such that $\operatorname{conv}(\mathcal W(A)) = \mathcal W(N)$. 
\end{thm}
\begin{proof}
For the converse, we note that $\mathcal W(N)$ is a polytope for a commuting tuple $N$. Thus, by assumption $\operatorname{conv}\mathcal W(A)$ is a polytope. Proposition \ref{prop:num_range_P_operators} implies moreover that $A$ are $\mathcal P$-operators for a polytope $\mathcal P$ if and only if $\operatorname{conv}(\mathcal W(A)) \subseteq \mathcal P$. Using the equivalent definition of $\mathcal P_{\min}$ from \cite{davidson2016dilations}, see Eq.~\eqref{eq:alternative-definition-min}, it follows thus that $A \in \mathcal P_{\min}$ for any polytope $\mathcal P$ such that $A$ are $\mathcal P$-operators, because the joint spectrum $\sigma(N) \subseteq \operatorname{conv}(\mathcal W(A))\subseteq \mathcal P$ for any such polytope. Hence, the $A$ are $\mathcal P$-compatible if and only if they are $\mathcal P$-operators.

For the remaining direction, it follows from Lemma \ref{lem:A-in-min} that $A$ being $\mathcal P$-compatible for all polytopes $\mathcal P$ such that the $A$ are $\mathcal P$-operators implies that $A \in (\operatorname{conv}(\mathcal W(A)))_{\min}$. Thus,
\begin{equation*}
    A = \sum_i z_i \otimes Q_i
\end{equation*}
 for a POVM $Q$ and $z_i \in \operatorname{conv}(\mathcal W(A))$. Naimark's dilation theorem implies that we can find pairwise-commuting Hermitian operators $N_j = \sum_i z_i(j) P_i$, where $P$ is the Naimark dilation of $Q$. By Remark \ref{rem:finite-dilations-and-POVMs-in-min}, we can assume that the sum is finite and that the $N_j$ are finite dimensional.  We can verify that $N$ is a dilation of $A$ and that $\{z_i\}_i=\sigma(N) \subseteq \mathcal W(N) \subseteq \operatorname{conv}(\mathcal W(A))$, because $\operatorname{conv}(\sigma(N)) = \mathcal W(N)$ for pairwise-commuting operators. On the other hand, as the $N$ are dilations of $A$, we have the reverse inclusion $\mathcal W(N) \supseteq \mathcal W(A)$. Noting that $\mathcal W(N)$ is convex as the tuple is pairwise commuting, it follows that $\operatorname{conv}\mathcal W(A) = \mathcal W(N)$.
\end{proof}

To show that Theorem \ref{thm:P-compatible-iff-P-operator} cannot be simplified at all, we would like to find a tuple $A \in \mathcal M_d^{\mathrm{sa}}(\mathbb C)^g$ which is not pairwise commuting itself, but which has a commuting dilation $N$ such that $\operatorname{conv}(\mathcal W(A)) = \mathcal W(N)$. We can in fact do better and characterize such $A$. We are indebted to the anonymous referee for providing us with this characterization, in particular for Lemma \ref{lem:equivalent-claim} and Theorem \ref{thm:A-almost-commuting}. In the following proofs, we will need the notion of Arveson extreme points for matrix convex sets (see \cite{evert2018extreme} for a discussion of different notions of extreme points in this setting). 
\begin{defi}
Let $\mathcal F \subseteq \bigsqcup_{n \in \mathbb N} M_n^{\mathrm{sa}}(\mathbb C)^g$ be a matrix convex set. Then, $A \in \mathcal F$ is an Arveson extreme point of $\mathcal F$ if for any dilation $X \in \mathcal F$, there exist a unitary $U$ and $Y \in \mathcal F$ such that
\begin{equation*}
    X = (U (A_1 \oplus Y_1) U^\ast, \ldots, U (A_g \oplus Y_g) U^\ast). 
\end{equation*}
\end{defi}
The next lemma is well-known (see for example, Section 3 of \cite{passer2022complex}), but for convenience we include a short proof.
\begin{lem} \label{lem:arveson-EP-min}
    Let $g \in \mathbb N$ and $\mathcal C \subset \mathbb R^g$ be a compact convex set. Then, any extreme point of $\mathcal C$ is an Arveson extreme point of $\mathcal C_{\min}$.
\end{lem}
\begin{proof}
    Let $n \in \mathbb N$ and let $X \in \mathcal C_{\min}(n)$ be a dilation of an extreme point $\lambda \in \mathcal C$. Hence, there exist an isometry $V: \mathbb C \hookrightarrow \mathbb C^n$ such that $\lambda = (V^\ast X_1 V, \ldots, V^\ast X_g V)$. We can identify $V$ with a vector $v \in \mathbb C^n$. As $X \in \mathcal C_{\min}(n)$, we can find a POVM $Q$ and points $z_i \in \mathcal C$ such that 
    \begin{equation*}
        X = \sum_i z_i \otimes Q_i .
    \end{equation*}
    Hence, 
    \begin{equation*}
        \lambda = (I \otimes V^\ast) X (I \otimes V) =  \sum_i v^\ast Q_i v z_i .
    \end{equation*}
    As $\lambda$ is an extreme point and $ (v^\ast Q_i v)_i $ is a probability distribution, $z_i = \lambda$ for at least one index $i$ and $v^\ast Q_j v = 0$ for all $j$ such that $z_j \neq \lambda$. For such $j$, the support of $Q_j$ is orthogonal to $v$. Thus,
    \begin{equation*}
        X = \lambda \otimes  v v^\ast+ \sum_{j: z_j \neq \lambda} z_j \otimes Q_j = (\lambda_1 \oplus Y_1, \ldots \lambda_g \oplus Y_g),
    \end{equation*}
where $Y = \sum_{j: z_j \neq \lambda} z_j \otimes Q_j$. Here, the direct sum is taken in a basis containing $v$. This proves the assertion.
 \end{proof}
 We will also use the matrix range of $A \in \mathcal M_d^{\mathrm{sa}}(\mathbb C)^g$, which is the image of $A$ under all unital completely positive (UCP) maps:
\begin{equation*}
    \operatorname{MR}(A)_n = \{(\Phi(A_1), \ldots, \Phi(A_g))~|~ \Phi: \mathcal M_d(\mathbb C) \to \mathcal M_n(\mathbb C)~\mathrm{UCP}\}.
\end{equation*}
 The matrix range $\sqcup_{n \in \mathbb N} \operatorname{MR}(A)_n$ of $A$ is then a closed bounded matrix convex set \cite[Proposition 2.5]{davidson2016dilations}. In this sense, $\operatorname{MR}(A)$ is the matrix convex set generated by $A$. With the help of the matrix range, we can give an equivalent characterization of the tuples $A$ we are looking for.
 \begin{lem} \label{lem:equivalent-claim}
Let  $d$, $g \in \mathbb N$ and let $A \in \mathcal M_d^{\mathrm{sa}}(\mathbb C)^g$. The following are equivalent:
\begin{enumerate}
    \item  $A$ admits a pairwise commuting dilation $N$ such that $\operatorname{conv}(\mathcal W(A)) = \mathcal W(N)$
    \item $\operatorname{MR}(A) = \operatorname{conv}(\mathcal W(A))_{\min}$.
\end{enumerate}
 \end{lem}
 \begin{proof}
For $(1) \implies (2)$, we can use \cite[Corollary 2.8]{davidson2016dilations} to infer that $\operatorname{conv}(\mathcal W(A))_{\min} = \operatorname{MR}(N)$. and using that $N$ is a dilation of $A$, moreover that $\operatorname{conv}(\mathcal W(A))_{\min} \supseteq \operatorname{MR}(A)$. The reverse inclusion follows since by definition  $\operatorname{MR}(A)(1) = \operatorname{conv}(\mathcal W(A))$.

For $(2) \implies (1)$, we realize that $A \in \operatorname{MR}(A)$, hence using the assumption $A \in \operatorname{conv}(\mathcal W(A))_{\min} $. Hence, there exists a commuting dilation $N$ with $\sigma(N) \subseteq \operatorname{conv}(\mathcal W(A))$, which implies $\mathcal W(N) \subseteq \operatorname{conv}(\mathcal W(A))$. The reverse inclusion follows as $N$ is a dilation of $A$.
 \end{proof}

 \begin{thm} \label{thm:A-almost-commuting}
     Let  $d$, $g \in \mathbb N$ and let $A \in \mathcal M_d^{\mathrm{sa}}(\mathbb C)^g$. The following are equivalent:
\begin{enumerate}
    \item $\operatorname{MR}(A) = \operatorname{conv}(\mathcal W(A))_{\min}$
    \item $A = (U (B_1 \otimes E_1) U^\ast, \ldots, U (B_g \otimes E_g) U^\ast)$, where $U$ is a unitary, $B$ normal pairwise commuting such that $\operatorname{conv}(\sigma(B)) = \operatorname{conv}(\mathcal W(A))$, and $E$ such that $\operatorname{MR}(E) \subseteq (\operatorname{conv}(\mathcal W(A)))_{\min}$.
\end{enumerate}
 \end{thm}
\begin{proof}
    For $(2) \implies (1)$, we infer from \cite[Corollary 4.4]{davidson2016dilations} that $\operatorname{MR}(B) = \operatorname{conv}(\mathcal W(A))_{\min}$. As moreover $\operatorname{MR}(E) \subseteq (\operatorname{conv}(\mathcal W(A)))_{\min}$ by assumption, it holds that $B \oplus E \in (\operatorname{conv}(\mathcal W(A)))_{\min}$ and hence also that $\operatorname{MR}(A) \subseteq (\operatorname{conv}(\mathcal W(A)))_{\min}$. The reverse inclusion is obvious since both sets agree on level $1$.

    For $(1) \implies (2)$, Lemma \ref{lem:arveson-EP-min} implies that all extreme points of $\operatorname{conv}(\mathcal W(A))$ are Arveson extreme points of $(\operatorname{conv}(\mathcal W(A)))_{\min}$. Note that the extreme points of $\operatorname{conv}(\mathcal W(A))$ have the form of $(v^\ast A_1 v, \ldots, v^\ast A_1 v)$ for some unit vector $v \in \mathbb C^g$. Hence, $A$ is a dilation of each of the extreme points of $\operatorname{conv}(\mathcal W(A))$. As these extreme points are Arveson extreme points, all of them appear as direct summands of $A$. Hence, $A$ is of the form $B \oplus E$ up to simultaneous unitary conjugation, where $B$ is the direct sum of the extreme points of $\operatorname{conv}(\mathcal W(A))$. Hence, it is diagonal and has joint spectrum equal to the extreme points of $\operatorname{conv}(\mathcal W(A))$. In particular, $\operatorname{conv}(\sigma(B)) = \operatorname{conv}(\mathcal W(A))$.  Clearly, we also have $\operatorname{MR}(E) \subseteq \operatorname{MR}(A)$ as $A$ is a dilation of $E$. This concludes the proof.   
\end{proof}

 \begin{cor} \label{cor:nicer-form}
Let  $d$, $g \in \mathbb N$ and let $A \in \mathcal M_d^{\mathrm{sa}}(\mathbb C)^g$ such that $0 \in \operatorname{int} (\operatorname{conv}\mathcal W(A))$. Then, $A$ is $\mathcal P$-compatible for \emph{all} polytopes $\mathcal P$ such that they are $\mathcal P$-operators if and only if it can be written as $A = (U (B_1 \otimes E_1) U^\ast, \ldots, U (B_g \otimes E_g) U^\ast)$, where $U$ is a unitary, $B$ normal pairwise commuting such that $\operatorname{conv}(\sigma(B)) = \operatorname{conv}(\mathcal W(A))$, and $E$ such that $\operatorname{MR}(E) \subseteq (\operatorname{conv}(\mathcal W(A)))_{\min}$.
 \end{cor}

 \begin{remark}
    We give an example of $A$ being $\mathcal P$-operators, but not $\mathcal P$-compatible, see \cite[Example 3.2]{li2020joint}. Take 
$$A_1:= \operatorname{diag}(1,1,-1,-1) \oplus \sigma_Z \quad \text{ and } \quad A_2:= \operatorname{diag}(1,-1,1,-1) \oplus \sigma_X.$$
Their joint numerical range is the square $[-1,1]^2$, while the matrices do not commute. The diagonal part is responsible for the square, while the off-diagonal part (which is non-commutative) has a unit disk as a joint numerical range, which is hidden by the larger square. This shows that having a numerical range which is a polytope is not equivalent to commuting. From joint measurability (see Section \ref{sec:meas-compatibility-polytope}), we know that $(A_1, A_2)$ are $[-1,1]^2$-operators, but not $[-1,1]^2$-compatible.

We can modify the construction slightly, however, to find an example $B$ being $\mathcal P$-compatible for all polytopes that $B$ are $\mathcal P$-operators.  Take 
$$B_1:= \operatorname{diag}(1,1,-1,-1) \oplus\frac{1}{\sqrt{2}}\sigma_Z \quad \text{ and } \quad B_2:= \operatorname{diag}(1,-1,1,-1) \oplus \frac{1}{\sqrt{2}}\sigma_X.$$ Now the the matrix range of $(\frac{1}{\sqrt{2}}\sigma_Z, \frac{1}{\sqrt{2}}\sigma_X)$ is contained in $([-1,1]^2)_{\min}$, as $(\sigma_Z, \sigma_X) \in ([-1,1]^2)_{\max}$ and $\frac{1}{\sqrt{2}} \in \Delta_{[-1,1]^2}(d)$ for all $d \in \mathbb N$ \cite{passer2018minimal}. We are thus in the setting of Corollary \ref{cor:nicer-form}.
\end{remark}

In view of Proposition \ref{prop:operator_compatible_max_min}, one might ask when $\mathcal P_{max}(d)=\mathcal P_{min}(d)$ holds. This is known to be the case if and only if $\mathcal P$ is a simplex. 

\begin{prop}[{\cite{fritz2017spectrahedral}}] \label{prop:min-max-simplex}
    Let $g$, $d \in \mathbb N$ and let $\mathcal P \subseteq \mathbb R^g$ be a polytope and $d \geq 2$. Then
    $$\mathcal P_{\min}(d) = \mathcal P_{\max}(d) \iff \text{$\mathcal P$ is a simplex.}$$
\end{prop}
\begin{proof}
    This follows from \cite[Theorem 4.3]{fritz2017spectrahedral} (see also \cite[Theorem 3.1]{huber2021note}).
\end{proof}
The statement even holds without the assumption that $\mathcal P$ is a polytope, see \cite[Corollary 2]{aubrun2021entangleability} together with \cite[Section 7]{passer2018minimal}.

\subsection{Additional results}

The notions of being $\mathcal P$-operators and $\mathcal P$-compatible behave well with respect to inclusion of polytopes. 

\begin{prop} \label{prop:polytope-inclusion}
    Let $\mathcal P \subseteq \mathcal Q$ be two polytopes with $0 \in \operatorname{int} \mathcal P$. If $A$ are $\mathcal P$-operators (resp.~$\mathcal P$-compatible), then $A$ are also $\mathcal Q$-operators (resp.~$\mathcal Q$-compatible).
\end{prop}

\begin{proof}
It follows from the definition of the maximal and minimal matrix convex that $\mathcal P \subseteq \mathcal Q$ implies both $\mathcal P_{\mathrm{min}} \subseteq \mathcal Q_{\mathrm{min}}$ and $\mathcal P_{\mathrm{max}} \subseteq \mathcal Q_{\mathrm{max}}$. The assertions then follow from Proposition \ref{prop:operator_compatible_max_min}.
\end{proof} 

\begin{ex}\label{ex:Paulis} We give an example of $A$ being $\mathcal Q$-compatible, but not $\mathcal P$-compatible for $\mathcal P \subseteq \mathcal Q$.
    Let $A= (\sigma_X, \sigma_Y, \sigma_Z) \in \mathcal M_2^{\mathrm{sa}}(\mathbb C)^3$ be the triplet of Pauli matrices. Then: 
    \begin{itemize}
        \item The joint numerical range of $A$ is the unit sphere of $\mathbb R^3$.
        \item $A$ are $[-1,1]^3$-operators. 
        \item $A$ are not $[-1,1]^3$-compatible, but are $\left( s_1^{-1} [-1,1]\right) \times \left( s_2^{-1} [-1,1]\right) \times \left( s_3^{-1} [-1,1]\right)$-compatible, for all $s_{1,2,3}>0$ with $s_1^2+s_2^2+s_3^2 \leq 1$.
    \end{itemize}
    The last point follows from Section \ref{sec:meas-compatibility-polytope} and known measurement compatibility results, see \cite{bluhm2018joint}, for example.
\end{ex}

Recall from Section \ref{sec:convexity} the definitions of the direct sum and the direct product of polytopes containing $0$. We gather next some results about the behavior of the Cartesian product and the direct sum operations when considering matrix levels.

\begin{prop} \label{prop:easy_consequences}
    Let $g_1$, $g_2 \in \mathbb N$. Let $A_i$ be $g_i$-tuples of operators, and $\mathcal P_i \in \mathbb R^{g_i}$ be polytopes with $0 \in \operatorname{int} \mathcal P_i$, $i=1,2$. Then:
    \begin{enumerate}
        \item $(A_1, A_2) \in (\mathcal P_1 \times \mathcal P_2)_{\max} \iff A_1 \in (\mathcal P_1)_{\max}$ and $A_2 \in (\mathcal P_2)_{\max}$.
        \item $(A_1, A_2) \in (\mathcal P_1 \times \mathcal P_2)_{\min} \implies A_1 \in (\mathcal P_1)_{\min}$ and $A_2 \in (\mathcal P_2)_{\min}$, but the converse does not hold in general.
        \item $A_1 \in (\mathcal P_1)_{\min}$ and $A_2 \in (\mathcal P_2)_{\min} \implies (A_1 \oplus 0, 0 \oplus A_2) \in (\mathcal P_1 \times \mathcal  P_2)_{\min}$ and $(A_1 \otimes I, I \otimes A_2) \in (\mathcal P_1 \times \mathcal  P_2)_{\min}$; this holds even for tuples of operators having different dimensions.
        \item If $q_1, q_2 \geq 0$ with $q_1+q_2 \leq 1$, then $A_1 \in (\mathcal P_1)_{\min}$ and $A_2 \in \mathcal (P_2)_{\min} \implies (q_1 A_1, q_2 A_2) \in ( \mathcal P_1 \times \mathcal P_2)_{\min}$.
        \item $(A_1, A_2) \in (\mathcal P_1 \oplus \mathcal P_2)_{\min} \implies A_1 \in (\mathcal P_1)_{\min}$ and $A_2 \in (\mathcal P_2)_{\min}$, but the converse does not hold in general.
        \item $A_1 \in (\mathcal P_1)_{\min}$ and $A_2 \in (\mathcal P_2)_{\min} \implies (A_1 \oplus 0, 0 \oplus A_2) \in (\mathcal P_1 \oplus \mathcal  P_2)_{\min}$; this holds even for tuples of operators having different dimensions.
        \item If $q_1, q_2 \geq 0$ with $q_1+q_2 \leq 1$, then $A_1 \in (\mathcal P_1)_{\min}$ and $A_2 \in \mathcal (P_2)_{\min} \implies (q_1 A_1, q_2 A_2) \in ( \mathcal P_1 \oplus \mathcal P_2)_{\min}$.
    \end{enumerate}
\end{prop}

\begin{proof}
For the first point, note that the condition $\langle (A_1, A_2), h_{12} \otimes \rho \rangle \leq 1$ has to be checked only for extreme points $h_{12} \in (\mathcal P_1 \times \mathcal P_2)^\circ =\mathcal  P_1^\circ \oplus \mathcal  P_2^\circ$. Such extreme points are either of the form $(h_1, 0)$ or of the form $(0,h_2)$, with $h_i \in \operatorname{ext} \mathcal P_i^\circ$, see \cite[Section 3.1]{bluhm2020compatibility}.

For the second point, let ${v^{(i)}_j}_{j \in [l_i]}$ be the extreme points of $\mathcal P_i$. Then, we can write
\begin{align*}
    (A_1,A_2) &= \sum_{i \in [l_1]} \sum_{j \in [l_2]} (v^{(1)}_i, v^{(2)}_j) \otimes C_{ij} \\
            &= \sum_{i \in [l_1]} (v^{(1)}_i, 0) \otimes C_{i}^{(1)} + \sum_{j \in [l_2]} (0, v^{(2)}_j) \otimes C_{j}^{(2)}
\end{align*}
with POVMs $C_{i}^{(1)} = \sum_{j \in [l_2]} C_{ij}$ and $C_{j}^{(2)} = \sum_{i \in [l_1]} C_{ij}$. Thus, 
\begin{equation*}
    A_i = \sum_{j \in [l_i]} v^{(i)}_j \otimes C_{j}^{(i)}.
\end{equation*}

For the third point, start with decompositions
\begin{align}
    \label{eq:A_1-min} A_1 &= \sum_{i=1}^k v_i \otimes C_i\\
    \label{eq:A_2-min} A_2 &= \sum_{j=1}^l w_j \otimes D_j,
\end{align}
with $\{v_i\}_{i \in [k]} \subset \mathcal P_1$ and $\{v_j\}_{j \in [l]} \subset \mathcal P_2$ and $\{C_i\}_{i \in [k]}$, $\{D_j\}_{j \in [l]}$ POVMs. From Eq.~\eqref{eq:direct-sum-smaller-Cartesian-product} we know that $\mathcal P_1 \oplus \mathcal P_2 \subseteq \mathcal P_1 \times \mathcal P_2$ and the first claim follows from the sixth point and Proposition \ref{prop:polytope-inclusion}. Finally, we can put 
\begin{equation*}
    (A_1 \otimes I, I \otimes A_2) = \sum_{i=1}^k \sum_{j=1}^l (v_i,w_j) \otimes (C_i \otimes D_j)
\end{equation*}
and note that $\{C_i \otimes D_j\}_{i \in [k], i \in [k]}$ forms a POVM.

The fourth point follows again from the seventh point, Eq.~\eqref{eq:direct-sum-smaller-Cartesian-product}, and Proposition \ref{prop:polytope-inclusion}.

For the fifth point, using again the general form of the extreme points of $P_1 \oplus P_2$, we have 
$$(A_1,A_2) = \sum_{i=1}^k (v_i,0) \otimes C_i + \sum_{j=1}^l (0,w_j) \otimes D_j,$$
for positive semidefinite operators $C_i,D_j$ such that $\sum_i C_i + \sum_j D_j \leq I_d$. Hence, both $C=(C_i)_{i \in [k]}$ and $D=(D_j)_{j \in [l]}$ are sub-POVMs and the conclusion follows from Proposition \ref{prop:sub-POVM-compatibility}. 

For the sixth point, start with the decompositions in from Eqs.~\eqref{eq:A_1-min}-\eqref{eq:A_2-min}. Then write
$$(A_1 \oplus 0, 0 \oplus A_2) = \sum_{i=1}^k (v_i,0) \otimes (C_i \oplus 0) + \sum_{j=1}^l (0,w_j) \otimes (0 \oplus D_j),$$
where $\{C_i \oplus 0\}_{i \in [k]} \sqcup \{0 \oplus D_j\}_{j \in [l]}$ forms itself a POVM.

For the seventh point, start again from Eqs.~\eqref{eq:A_1-min}-\eqref{eq:A_2-min}, and write
$$(q_1 A_1,q_2 A_2) = \sum_{i=1}^k (v_i,0) \otimes q_1C_i + \sum_{j=1}^l (0,w_j) \otimes q_2D_j,$$
where $\{q_1C_i\}_{i \in [k]} \sqcup \{q_2 D_j\}_{j \in [l]}$ forms itself a sub-POVM. Proposition \ref{prop:sub-POVM-compatibility} yields the conclusion.

Let us now show that the converses of the second and the fifth points do not hold in general. To this end, take $\mathcal P_{1,2} = [-1,1]$ and $A_1 = \sigma_X$, $A_2 = \sigma_Z$. Clearly, $\|A_{1,2}\| = 1$ so $A_i \in (\mathcal P_i)_{\max} = (\mathcal P_i)_{\min}$. Since the two observables $A_{1,2}$ are not compatible (see Example \ref{ex:Paulis} or the next section), we have that 
$$(A_1, A_2) \notin ([-1,1]^2)_{\min} = ([-1,1] \times [-1,1])_{\min} \supseteq ([-1,1] \oplus [-1,1])_{\min},$$
proving the claim.
\end{proof}
\begin{remark}
    The results above generalize in the obvious manner to more than two polytopes and tuples of operators. 
\end{remark}

\section{Measurement compatibility} \label{sec:meas-compatibility-polytope}
The first example we will consider concerns the compatibility of dichotomic quantum measurements. Remember that in this case, we identify the POVM $\{E_i, I-E_i\}$ simply with the effect $E_i$. Our aim is to rephrase the compatibility of $g$ dichotomic measurements in dimension $d$ as $\mathcal P$-compatibility with $\mathcal P = [-1,1]^g$. This recovers results from \cite{Bluhm2022norms} with alternative proofs. Taking a different $\mathcal P$, the results will extend to measurements with more outcomes. We will discuss this at the end of this section.

\begin{prop}\label{prop:effects-measurement}
    Let $g$, $d \in \mathbb N$ and let $A \in \mathcal (M_d^\mathrm{sa}(\mathbb C))^g$. Then, the $A$ are $[-1,1]^g$-operators if and only if the $E_i=\frac{1}{2}(A_i + I_d)$ are effects for all $i \in [g]$.
\end{prop}
\begin{proof}
    By Proposition \ref{prop:num_range_P_operators}, the $A$ are $[-1,1]^g$-operators if and only if 
    \begin{equation*}
        \operatorname{Tr}[A_i \rho] \in [-1,1] \qquad \forall \rho \in \mathcal M_d^{1,+}(\mathbb C), \forall i \in [g].
    \end{equation*}
    This is equivalent to $-I_d \leq A_i \leq I_d$ $\forall i \in [g]$, from which the assertion follows.
\end{proof}

\begin{prop}\label{prop:effects-compatible}
    Let $g$, $d \in \mathbb N$ and let $A \in \mathcal (M_d^\mathrm{sa}(\mathbb C))^g$. Then, the $A$ are $[-1,1]^g$-compatible if and only if the $E_i=\frac{1}{2}(A_i + I_d)$ are compatible effects for all $i \in [g]$.
\end{prop}
\begin{proof}
    We note that the extreme points of $[-1,1]^g$ are the sign vectors $\epsilon_j\in \{\pm 1\}^g$, $j \in [{2^g}]$. From Definition \ref{defi:P-operators-P-compatible}, the $A$ are $[-1,1]^g$-compatible if and only if
    \begin{equation*}
        A_i = \sum_{\epsilon \in \{\pm 1\}^g} \epsilon(i) C_\epsilon \qquad\forall i \in [g]
    \end{equation*}
    and for some POVM $\{C_\epsilon\}_{\epsilon \in \{\pm 1\}^g}$. As $ \sum_{\epsilon \in \{\pm 1\}^g} C_\epsilon = I_d$ and $\frac{1}{2}(\epsilon(i)+1) \in \{0,1\}$, we obtain
    \begin{equation*}
        E_i = \sum_{\substack{\epsilon \in \{\pm 1\}^g\\\epsilon(i)=1}} C_\epsilon \quad\qquad \forall i \in [g].
    \end{equation*}
    Thus, $\{C_\epsilon\}_{\epsilon \in \{\pm 1\}^g}$ is a joint POVM for the measurements defined by $\{E_i\}_{i \in [g]}$.
\end{proof}

In fact, this example has motivated our terminology of $\mathcal P$-operators and $\mathcal P$-compatibility.

\begin{remark}\label{rmk:easy-consequences}
    We can extract a few easy consequences about compatible measurements from our theory of $\mathcal P$-compatibility. Most of them are easy to check directly and should be seen primarily as sanity checks and as providing a way to think about the propositions from which they follow in a more abstract setting.

Point (1) of Proposition \ref{prop:easy_consequences} implies that a collection of matrices is a collection of effects if and only if each matrix is an effect individually. 

Point (2) of Proposition \ref{prop:easy_consequences} implies that if you have a collection of compatible effects, taking any subset of these effects is still a collection of compatible effects. The proof of this point shows that the corresponding joint POVM arises from taking marginals of the joint POVM for the larger collection of compatible effects. It is known that you cannot simply combine sets of compatible effects to result in a larger set of compatible effects. This is easy to see, since any effect is compatible with itself, such that there would not be incompatible effects otherwise.

Point (3) of Proposition \ref{prop:easy_consequences} implies that if you insert two collections of compatible measurements into different blocks of a larger matrix, their union is compatible. The same is true if you combine measurements on different subsystems.

Point (4) of Proposition \ref{prop:easy_consequences} implies that for two collections of compatible measurements, adding a certain amount of noise to each of them implies that their union remains compatible. More concretely, if you have, for example, two collections of dichotomic measurements $\{E_1, \ldots, E_g\}$ and $\{F_1, \ldots, F_{g^\prime}\}$, where the measurements $\{E_1, \ldots, E_g\}$ are compatible and the measurements $\{F_1, \ldots, F_{g^\prime}\}$ are compatible, then the collection of noisy measurements 
\begin{equation*}
    \{q_1 E_1 + (1-q_1)I/2, \ldots, q_1 E_g + (1-q_1)I/2, q_2 F_1 + (1-q_2)I/2, \ldots, q_2 F_{g^\prime} + (1-q_2)I/2\}
\end{equation*}
are compatible as well if $q_1 + q_2 \leq 1$, since $2(q_1 E_1 + (1-q_1)I/2) - I = q_1 (2 E_i - I)$. The latter shows that adding noise with parameter $q_1$ to the measurements is equivalent to multiplying the corresponding tensor $A$ as in Proposition \ref{prop:effects-measurement} by $q_1$. This can be interpreted in terms of coin tossing and mixing as in \cite[Section 2.3]{Heinosaari2016}.

From Proposition \ref{prop:min-max-simplex}, it follows that there is no incompatibility if one considers only a single POVM. Conversely, the proposition implies that if the dimension is at least $2$, we can always find a collection of incompatible measurements if we consider at least two measurements with at least two outcomes.
\end{remark}

\begin{remark}
    We could have also used Proposition \ref{prop:factoring} for the proofs in this section, since any collection of $g$ measurements $E^{(i)}$ with $k_i$ outcomes each, $i \in [g]$, can be seen as a measurement map $\mathcal E$ from $\mathcal M_d^{1,+}(\mathbb C) \to \Delta_{\mathbf k}$, given as
    \begin{equation*}
        \mathcal E: \rho \mapsto (\operatorname{Tr}[E^{(i)}_1 \rho] \delta^{(i)}_1 + \ldots + \operatorname{Tr}[E^{(i)}_1\rho] \delta_{k_i}^{(i)})_{i \in [g]}.
    \end{equation*}
    Here, $\Delta_{\mathbf k}$ is the polysimplex, i.e., the GPT generated by the Cartesian product of simplices $\Delta_{k_1} \times \ldots \times \Delta_{k_g}$, and $\delta^{(i)}_1, \ldots,  \delta_{k_i}^{(i)}$ are the vertices of the simplex $\Delta_{k_i}$. We refer the reader to \cite{jencova2018incompatible} for details. In Theorem 1 of \cite{jencova2018incompatible}, it was shown that the measurements $E^{(i)}$ are compatible if and only if $\mathcal E$ factors through a simplex (see also \cite{Bluhm2022GPT}). This fact could also have been proven from Proposition \ref{prop:factoring} combined with Propositions \ref{prop:effects-measurement} and \ref{prop:effects-compatible}. 
    
    Finally, we could have proven Propositions \ref{prop:effects-measurement} and \ref{prop:effects-compatible} using the techniques in \cite{Bluhm2022norms} based on tensor norms on Banach spaces and their link to matrix convex sets.
\end{remark}

    In this example, we have focused on dichotomic measurements for simplicity, but the correspondence works for general POVMs. Let $k \in \mathbb N$. Defining 
    \begin{equation}\label{eq:def-simplex}
        \mathcal P_{k} = \{x \in \mathbb R^{k-1}: \langle -k e_j, x \rangle \leq 1~\forall j \in [k-1], \langle k(1, \ldots, 1), x \rangle \leq 1\},
    \end{equation}
    where the $e_j$ are the standard basis vectors, we can check that $\mathcal P_{k}$ is a polytope with extreme points 
    \begin{equation*}
        \left\{-\frac{1}{k}(1, \ldots, 1) +  e_j~\forall j \in [k]\right\} \cup \left\{-\frac{1}{k}(1, \ldots, 1)\right\}
    \end{equation*}
Then, we obtain the following statement that generalizes Propositions \ref{prop:effects-measurement} and \ref{prop:effects-compatible}:
\begin{prop}
    Let $g$, $d$, $k_i \in \mathbb N$ for all $i \in [g]$ and let $A \in (\mathcal M_d^\mathrm{sa}(\mathbb C))^{k_1 + \ldots + k_g - g}$. Let 
    $E_i^{(j)}=A_{k_1 + \ldots + k_{j-1} + i-j +1} + \frac{1}{k_j}I_d$ for all $i \in [k_j - 1]$ and  $E_{k_j}^{(j)} = I_d -  E_1^{(j)} - \ldots -  E_{k_j-1}^{(j)}$ for all $j \in [g]$. Let $\mathcal P_{\mathbf k}:=\mathcal P_{k_1} \times \ldots \times \mathcal P_{k_g}$. Then,
    \begin{enumerate}
        \item The $A$ are $\mathcal P_{\mathbf k}$-operators if and only if the tuples $(E^{(j)}_1, \ldots, E^{(j)}_{k_j})$ are POVMs for all $j \in [g]$.
        \item The $A$ are $\mathcal P_{\mathbf k}$-compatible if and only if the tuples $(E^{(j)}_1, \ldots, E^{(j)}_{k_j})$ are compatible POVMs for all $j \in [g]$.
    \end{enumerate}
\end{prop}
The proof proceeds analogously to the proofs of Propositions \ref{prop:effects-measurement} and \ref{prop:effects-compatible}. The statements of Remark \ref{rmk:easy-consequences} carry over with the necessary adjustments. We note that $\mathcal P_{k}$ is the polar of the matrix jewel base at the level $1$, called  $\mathcal D_{\jewel,k}(1)$ in \cite{bluhm2020compatibility}; note that there is a difference in normalization between the definitions above and the ones in \cite{bluhm2020compatibility} by a factor of $2$. This set was denoted as $\mathcal D_{\cuboid,k}(1)$ in that paper. Furthermore, $\mathcal P_{\mathbf k}$ is the polar of the matrix jewel at level $1$, called $\mathcal D_{\jewel, \mathbf{k}}(1)$ in \cite{bluhm2020compatibility}. This set was denoted as $\mathcal D_{\cuboid,  \mathbf{k}}(1)$ in that paper. Contrary to the dichotomic case this proposition does not follow from \cite{Bluhm2022norms}, since we cannot define tensor norms corresponding to these matrix convex sets as they are asymmetric.

\section{Magic squares}\label{sec:magic}

In this section we discuss the case of \emph{magic squares} by relating them to the Birkhoff polytope. This is one of the main examples we discuss in detail, and was the starting point of our investigation.

We recall the following definitions from \cite{cuevas2020quantum} and from the quantum groups literature \cite{banica2007quantum}, restricting ourselves to the matrix algebra setting.

\begin{defi}\label{def:magic-square}
Let $d$, $N \in \mathbb N$. A block matrix $A \in \mathcal M_N(\mathcal M_d(\mathbb C))$ having positive semidefinite blocks $A_{ij} \geq 0$ $\forall i,j \in [N]$ is called 
\begin{itemize}
    \item a \emph{quantum magic square} if it is \emph{bistochastic}: 
    $$\forall i \in [N], \, \sum_{j=1}^N A_{ij} = I_d \quad \text{ and } \quad \forall j \in [N], \, \sum_{i=1}^N A_{ij} = I_d;$$
    \item a \emph{semiclassical magic square} if there exist a POVM $(Q_\pi)_{\pi \in \mathcal S_N}$ such that
    $$A = \sum_{\pi \in \mathcal S_N} P_\pi \otimes Q_\pi,$$
    where $P_\pi$ is the permutation matrix associated to $\pi$.
\end{itemize} 
\end{defi}

\begin{remark}\label{rmk:semiclassical-is-magic}
It is easy to verify that a semiclassical magic square is a magic square. Indeed, $P_\pi$ being a permutation matrix means that $(P_\pi)_{ij} = \delta_{\pi(i), j}$. Thus,
\begin{equation*}
    A_{ij} = \sum_{\pi \in \mathcal S_n} \delta_{\pi(i), j} Q_\pi,
\end{equation*}
such that $A_{ij} \geq 0$. Moreover, for any $i \in [N]$, using that the $Q_\pi$ form a POVM,
\begin{equation*}
    \sum_{j = 1}^N   A_{ij} = \sum_{\pi \in \mathcal S_n}  \sum_{j = 1}^N \delta_{\pi(i), j} Q_\pi = \sum_{\pi \in \mathcal S_n} Q_\pi = I_d.
\end{equation*}
Likewise, for any $j \in [N]$,
\begin{equation*}
    \sum_{i = 1}^N   A_{ij} = \sum_{\pi \in \mathcal S_n}  \sum_{i = 1}^N \delta_{\pi(i), j} Q_\pi = \sum_{\pi \in \mathcal S_n}  \sum_{i = 1}^N \delta_{i, \pi^{-1}(j)} Q_\pi = \sum_{\pi \in \mathcal S_n} Q_\pi = I_d.
\end{equation*}
\end{remark}

\subsection{The Birkhoff polytope}\label{sec:Birkhoff}

In this subsection we shall introduce and study the basic properties of matrix convex sets built upon the celebrated Birkhoff polytope. These objects will be related to (semiclassical) magic squares in Subsection \ref{sec:correspondence}.

The \emph{Birkhoff polytope} $\mathrm{Birk}_N$ is the convex set of bistochastic matrices from $\mathcal M_N(\mathbb R)$ \cite[Example 0.12]{ziegler2012lectures}. Famously, its vertices are the permutation matrices $P_\pi$, for $\pi \in \mathcal S_N$. Since this polytope lives in an $(N-1)^2$ dimensional affine hyperplane of $\mathbb R^{N^2}$, we shall consider the convex body version of the Birkhoff polytope, defined as follows. For an arbitrary matrix $X \in \mathcal M_N(\mathbb C)$, we shall denote by $X^{(N-1)}$ the principal submatrix obtained by deleting the last row and the last column from $X$. We consider the convex body obtained by truncating bistochastic matrices, after centering the Birkhoff polytope at $J/N$, where $J$ is the matrix in which all entries are $1$. 

\begin{defi}
For a given $N \geq 2$, the \emph{Birkhoff body} $\mathcal B_N$ is  defined as the set of $(N-1) \times (N-1)$ truncations of $N \times N$ bistochastic matrices, shifted by $J_{N-1}/N$:
$$\mathcal B_N = \{A^{(N-1)} - J_{N-1}/N \, : \, A \in \mathcal M_N(\mathbb R) \, \text{ bistochastic}\} \subset \mathcal M_{N-1}(\mathbb R) \cong \mathbb R^{(N-1)^2}.$$
\end{defi}
For example, it is easy to see that $\mathcal B_2 = [-1/2, 1/2]$. Let us now introduce an important notation: to a matrix $X \in \mathcal M_{N-1}(\mathcal M_d(\mathbb C))$, we associate the matrix $\tilde X \in \mathcal M_{N}(\mathcal M_d(\mathbb C))$ given by
$$\tilde X_{ij} = \frac{I_d}{N}  + \begin{cases}
X_{ij}, &\qquad \text{ if } i,j \in [N-1]\\
-\sum_{k=1}^{N-1}X_{ik}, &\qquad \text{ if } i \in [N-1], \, j = N\\
-\sum_{k=1}^{N-1}X_{kj} , &\qquad \text{ if } j \in [N-1], \, i = N\\
\sum_{k,l=1}^{N-1}X_{kl}, &\qquad \text{ if } i,j=N.
\end{cases}$$
The matrix $\tilde X$ agrees with $X+J_{N-1}/N\otimes I_d$ on the top $(N-1) \times (N-1)$ corner, and has row and column sums equal to $I_d$.

The convex and combinatorial properties of the Birkhoff polytope have been studied in a series of papers by Brualdi and Gibson, see \cite{brualdi1977convex}.

\begin{prop} \label{prop:birkhoff-body}
Let $N\in \mathbb N$. The Birkhoff body $\mathcal B_N$ has $N!$ extreme points $P_\pi^{(N-1)}-J_{N-1}/N$ and is described by the following inequalities:
\begin{align}
\label{eq:BN-face-elt}    \forall i,j \in [N-1], &\quad A_{ij} \geq -1/N\\
\label{eq:BN-face-row}    \forall i \in [N-1], &\quad \sum_{j=1}^{N-1} A_{ij} \leq 1/N\\
\label{eq:BN-face-col}    \forall j \in [N-1], &\quad \sum_{i=1}^{N-1} A_{ij} \leq 1/N\\
\label{eq:BN-face-sum}    &\quad \sum_{i,j=1}^{N-1} A_{ij} \geq -1/N.
\end{align}
For $N \geq 3$, its $N^2$ facets are given by replacing one of the $N^2$ inequalities by an equality. For $N=2$, there are only $2$ facets given by $A_{11} = -1/2$ and $A_{11} = 1/2$, respectively. Finally, it has $0$ in its interior.
\end{prop}
\begin{proof}
We build on the findings on the extreme points and facets of the Birkhoff polytope in \cite{brualdi1977convex}. First, let us point out that, given the $(N-1) \times (N-1)$ shifted truncation of a $N \times N$ bistochastic matrix, there is a unique way to recover the bistochastic matrix, given by the $X \mapsto \tilde X$ mapping defined above. This fact settles the claim about extreme points. Regarding the facets, note that the second and the third type of inequalities above ensure that one can fill the last element in the first $(N-1)$ rows and columns with a non-negative number. The last inequality ensures the non-negativity of the bottom-right entry of $\tilde A$. The fact that $0 \in \operatorname{int} \mathcal B_N$ holds, since $0$ fulfills all inequalities strictly.
\end{proof}

\subsection{Connecting magic squares to matrix convex sets}\label{sec:correspondence}

This subsection contains the main insight of this section: there is an intimate relation between the Birkhoff body $\mathcal B_N$ introduced in Section \ref{sec:Birkhoff} and the (semiclassical property for) magic squares introduced in Section \ref{sec:magic}. 

\begin{thm}\label{thm:min-max-magic-square}
Let $d$. $N \in \mathbb N$ Consider $A \in \mathcal M_d^{\mathrm{sa}}(\mathbb C)^{(N-1)^2}$ and the corresponding matrix $\tilde A \in \mathcal M_N(\mathcal M_d(\mathbb C))$. Then:
\begin{enumerate}
    \item the matrix $\tilde A$ is a magic square if and only if $A \in (\mathcal B_N)_{\max}$;
    \item the matrix $\tilde A$ is a semiclassical magic square if and only if $A \in (\mathcal B_N)_{\min}$.
\end{enumerate}
\end{thm}
\begin{proof}
Requiring that the tuple $A$ satisfies the inequalities from Proposition \ref{prop:birkhoff-body} is easily seen to be equivalent to the fact that $\tilde A$ is a magic square, establishing the first point. Similarly, the second point follows from the form of the extreme points of the Birkhoff body, which are related to permutation matrices. Indeed, for $i,j \in [N-1]$, we have 
$$A_{ij} = \sum_{\pi \in \mathcal S_N} (P_\pi(i,j) - 1/N)Q_\pi$$
for a POVM $Q = (Q_\pi)_{\pi \in \mathcal S_N}$. In turn, we have 
$$\tilde A_{ij} = I/N + \sum_{\pi \in \mathcal S_N} (P_\pi(i,j) - 1/N)Q_\pi = \sum_{\pi \in \mathcal S_N} P_\pi(i,j) Q_\pi.$$
For $j=N$ and $i \in [N-1]$, we have 
\begin{align*}
    \tilde A_{iN} &= I - \sum_{j =1}^{N-1} \tilde A_{ij} = I - \sum_{j =1}^{N-1}\sum_{\pi \in \mathcal S_N} P_\pi(i,j) Q_\pi\\
    &= \sum_{\pi \in \mathcal S_N} \left( 1 - \sum_{j =1}^{N-1}P_\pi(i,j) \right) Q_\pi = \sum_{\pi \in \mathcal S_N} P_\pi(i,N) Q_\pi.
\end{align*}
Similar computations yield the cases $i=N, j \in [N-1]$ and $i=j=N$, finishing the proof. 
\end{proof}

\subsection{Semiclassicality vs.~compatibility}\label{sec:counterexample}

The magic square condition from Definition \ref{def:magic-square} can be equivalently stated that the columns, respectively the rows of the square form quantum measurements: 
$$R^{(i)}:=\{A_{ij}\}_{j \in [N]} \qquad \text{ and } \qquad C^{(j)} := \{A_{ij}\}_{i \in [N]}. $$

We have the following observation, see also \cite{guerini2018joint}. 
\begin{prop}
    Let $N \in \mathbb N$. For $i$, $j \in [N]$, the $2N$ measurements $R^{(i)}$ and $C^{(j)}$ (defined as above) of a \emph{semiclassical} magic square are compatible. 
\end{prop}
\begin{proof}
    The column and row POVMs are post-processings in the sense of Lemma \ref{lem:post-processing} of the measurement $Q_\pi$ from Definition \ref{def:magic-square}. For example, 
    $$C^{(j)}_i = A_{ij} = \sum_{\pi \in \mathcal S_N} \delta_{\pi(i),j} Q_\pi.$$
    Above, we have used the post-processing
    $$p_\pi(i | C^{(j)}) = \delta_{\pi(i),j} = \delta_{i,\pi^{-1}(j)}.$$
    Similar expressions can be written for the row measurements, finishing the proof. 
\end{proof}

We now show, via an example, that the converse to the proposition above does not hold, i.e., there exist magic squares with compatible measurements $C^{(i)}$ and $R^{(j)}$ without the magic square being semiclassical. Consider the magic square in Table \ref{tab:counterexample}. Here, $e_1$, $e_2$ are the standard basis vectors in $\mathbb C^2$ and $f_1 = 1/\sqrt{2}(e_1 + e_2)$, $f_2 = 1/\sqrt{2}(e_1 - e_2)$.

\begin{table}[htb]
    \centering
    \begin{tabular}{|c|c|c|c|}
         \hline
         $\frac{1}{2}e_1 e_1^\ast$ & $\frac{1}{2}e_2 e_2^\ast$ & $0$ & $\frac{1}{2}I_2$  \\ 
         \hline
         $\frac{1}{2}e_2 e_2^\ast$ & $\frac{1}{2}e_1 e_1^\ast$ & $\frac{1}{2}I_2$ & $0$ \\
         \hline
         $0$ & $\frac{1}{2}I_2$ &    $\frac{1}{2}f_1 f_1^\ast$ & $\frac{1}{2}f_2 f_2^\ast$ \\
         \hline
         $\frac{1}{2}I_2$ & $0$ &    $\frac{1}{2}f_2 f_2^\ast$ & $\frac{1}{2}f_1 f_1^\ast$\\
         \hline
    \end{tabular}
    \caption{Example for $8$ qubit POVMs arranged in a magic square which are compatible, but which do not form a semiclassical magic square.}
    \label{tab:counterexample}
\end{table}
It is easy to see that the measurements in the rows and columns of Table \ref{tab:counterexample} in fact reduce to only two different POVMs (since all the other ones arise only as relabeling of outcomes). These two POVMs are
\begin{equation} \label{eq:counterexample_POVMs}
    \left( \frac{1}{2}e_1 e_1^\ast, \frac{1}{2}e_2 e_2^\ast, \frac{1}{2}I_2, 0\right) \quad \mathrm{and} \quad \left( \frac{1}{2}f_1 f_1^\ast, \frac{1}{2}f_2 f_2^\ast, \frac{1}{2}I_2, 0\right) 
\end{equation}
Moreover, it is straightforward to verify that 
\begin{equation*}
    \left( \frac{1}{2}e_1 e_1^\ast, \frac{1}{2}e_2 e_2^\ast, \frac{1}{2}f_1 f_1^\ast, \frac{1}{2}f_2 f_2^\ast\right)
\end{equation*}
is a joint POVM from which the POVMs in Eq.~\eqref{eq:counterexample_POVMs} arise via classical post-processing. Hence, all the POVMs shown in the rows and columns of Table \ref{tab:counterexample} are compatible.

It is a bit tedious to show that the magic square in Table \ref{tab:counterexample} is not semiclassical: if it was semiclassical, it would have an expression as
\begin{equation*}
    \sum_{\pi \in \mathcal S_4} P_{\pi} \otimes Q_\pi, \qquad Q_\pi \geq 0~\forall \pi  \in \mathcal S_4,\qquad ~\sum_\pi Q_\pi = I_2
\end{equation*}
where $P_\pi = (\delta_{\pi(i),j})_{i, j \in [4]}$ are the permutation matrices representing $\pi$. Let us denote by $A_{i,j}$ the entries of the magic square. Since $A_{1,3} = A_{2,4} = A_{3,1} = A_{4,2}=0$, only the $Q_\pi$ with $\pi(1)\neq3, \pi(2)\neq 4, \pi(3) \neq 1, \pi(4) \neq 2$ can be nonzero. Thus out of the 24 $Q_\pi$, $15$ have to be zero. Since $A_{1,1}, A_{1,2}, A_{2,1}, A_{2,2}, A_{3,3}, A_{3,4}, A_{4,3},, A_{4,4}$ are rank one, it follows that $Q_\pi$ summing to these entries have to be proportional to the corresponding rank one projectors. For example $Q_{(1,2,3,4)}$ has to be proportional to both $e_1e_1^\ast$ and $f_1f_1^\ast$ and therefore must be $0$. This reasoning sets another $8$ of the $Q_\pi$ to $0$. The only non-zero element is thus $Q_{(4,3,2,1)}$, which needs to be $I/2$. However, this choice of $Q_\pi$ does not generate the required magic square and contradicts the requirement of $Q$ being a POVM.

Clearly, the measurements in Table \ref{tab:counterexample} are not linearly independent, so we are not in the setting of \cite[Corollary 5]{guerini2018joint}. Similarly, the post-processing map used here is the ``marginal'' one, which is not symmetric in the sense of \cite[Proposition 3]{guerini2018joint}; these considerations show that the counterexample above do not contradict the results in \cite{guerini2018joint}.

Using our Proposition \ref{prop:factoring}, we can recover the characterization in \cite[Theorem 4]{guerini2018joint} of semiclassical magic squares as compatibility plus symmetric post-processing:
\begin{prop}
    $A \in (\mathcal B_N)_{\min}$ if and only if the tuples $C^{(j)}:=\{\tilde A_{ij}\}_{i \in [N]}$, $R^{(i)}:=\{\tilde A_{ij}\}_{j \in [N]}$ where $i$, $j \in [N]$ are compatible POVMs with the post-processing in Lemma \ref{lem:post-processing} satisfying $p_\lambda(i|R^{(j)}) = p_\lambda(j|C^{(i)})$ for all $i$, $j \in [N]$ and all $\lambda$.
\end{prop}
\begin{proof}
    Following the strategy for the proof of Theorem \ref{thm:min-max-magic-square}, using Proposition \ref{prop:factoring}, one can verify that being a semiclassical magic square is equivalent to the map $\tilde{\mathcal A}: \mathcal M_d^{1,+}(\mathbb C) \to \mathrm{Birk}_N$,
    \begin{equation*}
        \tilde{\mathcal A}(X) = \mathrm{Tr}[\tilde A_{ij}X]_{i,j \in [N]},
    \end{equation*}
    factoring through a $k$-simplex. The factorization is equivalent to the existence of a POVM $Q = \{Q_i\}_{i \in [k]}$ and a map $\nu: \Delta_k \to \mathrm{Birk}_N$ such that $\nu(\delta_l) = z_l$, where the $\delta_l$ are the vertices of the simplex for $l \in [k]$. Then, we can identify $p_\lambda(i|C^{(j)}) = \nu(\delta_\lambda)(i,j) = p_\lambda(j|R^{(i)})$ if we see $z_\lambda$ as a bistochastic matrix with entries $z_\lambda(i,j)$. It is easy to check that the $p_\lambda(j|R^{(i)})$, $p_\lambda(i|C^{(j)})$ are conditional probabilities because the $z_\lambda$ are bistochastic matrices and that 
    \begin{equation*}
        A_{ij} = \sum_{\lambda \in [k]} z_\lambda(i,j) Q_\lambda = \sum_{\lambda \in [k]} p_\lambda(i|C^{(j)}) Q_\lambda= \sum_{\lambda \in [k]} p_\lambda(j|R^{(i)}) Q_\lambda.
    \end{equation*}
    Since $C^{(j)}_i = A_{ij} = R^{(i)}_j$ and our identification between conditional probabilities and bistochastic matrices can be reversed, the assertion follows.
\end{proof}

\section{POVMs with common elements} \label{sec:POVM-common-elements}

We discuss in this section a very general way of defining polytopes for which the notions of $\mathcal P$-operators and $\mathcal P$-compatibility have a very clear physical interpretation. Interestingly, the mathematical framework is that of (measurement) hypergraphs, which is precisely the one used in the combinatorial approach to contextuality \cite{acin2015combinatorial}. We would like to point out that, beyond the mathematical correspondence, there is no clear physical relation between contextuality and the notion of compatibility discussed here.

We start by recalling that a hypergraph is a pair $G=(V,E)$, where $V$ is a non-empty set and $E$ is a set of non-empty subsets $\emptyset \neq e \subseteq V$, called \emph{hyperedges}. In this section, we shall assume that our hypergraphs have \emph{no isolated vertices}, i.e.
$$\forall v \in V \quad \exists e \in E \, \text{ s.t. } \, v \in e \iff \bigcup_{e \in E} e= V.$$

\begin{defi}\label{def:probability-hypergraph}
A hypergraph $G$ (with no isolated vertices) is called a \emph{probability hypergraph} if there exists a function $\pi : V \to (0,1]$ such that 
\begin{equation}\label{eq:probability-hypergraph-normalization}
    \forall e \in E, \quad \sum_{v \in e} \pi(v) = 1.
\end{equation}
In other words, given a probability hypergraph, we can associate to each vertex a positive number in such a way that all hyperedges correspond to probability vectors. We denote by $\Pi(G)$ the set of all functions $\pi : V \to [0,1]$ such that Eq. \eqref{eq:probability-hypergraph-normalization} holds. 
\end{defi}

Note that we need a function $\pi : V \to (0,1]$ in the definition of probability hypergraph, whereas the set $\Pi(G)$ consists of functions $\pi : V \to [0,1]$.

We display in Figure \ref{fig:hypergraphs} three hypergraphs. The first two are probability hypergraphs, since assigning weight $\pi(v) = 1/2$ to each vertex of the first graph and $\pi(v) = 1/3$ to each vertex of the second satisfies the condition from Definition \ref{def:probability-hypergraph}. However, the third one is \emph{not} a probability hypergraph: hyperedge $\{1\}$ implies $\pi(1) = 1$, while the other two hyperedges $\{1,2\}$ and $\{1,3\}$ force $\pi(2) = \pi(3) = 0$, contradicting $\pi(v) >0$.

\begin{figure}[htb]
    \centering
    \includegraphics{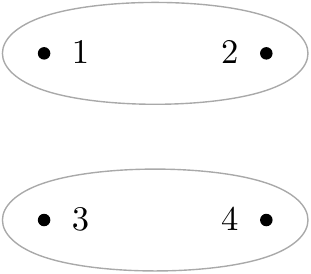} \qquad\qquad
    \includegraphics{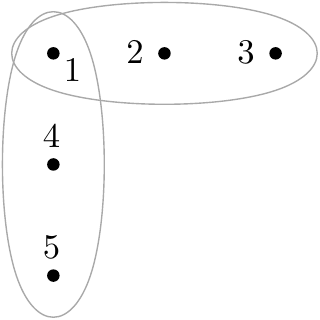} \qquad\qquad
    \includegraphics{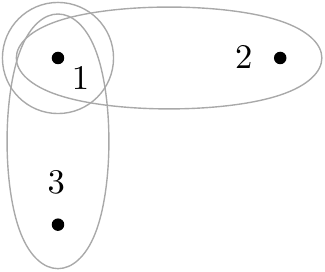} 
    \caption{Three hypergraphs. The first two are probability hypergraphs, while the third one is not.}
    \label{fig:hypergraphs}
\end{figure}

The incidence relation between vertices and hyperedges generate an equivalence relation on $E$, as the closure of the binary relation
$e \cap e' \neq \emptyset \implies e \sim e'.$

The following result establishes a large family of hypergraphs as probability hypergraphs. 

\begin{prop}
    All hypergraphs $G = (V,E)$ having the property that edges in the same incidence equivalence class have the same cardinality
    $$e \sim e' \implies |e| = |e'|$$
    are probability hypergraphs. 
\end{prop}
\begin{proof}
    The result follows by setting 
    $$\forall v \in V, \qquad \pi(v) := \frac{1}{|e|} \quad \text{ for } v \in e.$$
    The assignment above is independent of the choice of the hyperedge $e \ni v$ by the hypothesis. Clearly, $\pi \in \Pi(G)$ and $\pi(v)>0$, proving the claim.
\end{proof}

In particular, hypergraphs with disjoint hyperedges ($e \neq e' \implies e \cap e' = \emptyset$) are probability hypergraphs. The first two examples in Figure \ref{fig:hypergraphs} satisfy the assumptions of the proposition above. 

\bigskip 

In what follows, we shall associate to a probability hypergraph $G$ a polytope in an essentially unique manner. To start, consider the set 
$$\Pi^0(G) := \{ \pi : V \to \mathbb R \, : \, \forall e \in E, \, \sum_{v \in e} \pi(v) = 0\}.$$
Clearly, $\Pi^0(G)$ is a vector space, and we set $g := \dim \Pi^0(G)$.

Let $G$ be a probability hypergraph, $\pi_* \in \Pi(G)$, and consider a basis $\pi_1, \ldots, \pi_g$ of the vector space $\Pi^0(G)$. Define the set
    $$\mathcal P := \{a \in \mathbb R^g \, : \, \pi_* + \sum_{x=1}^g a_x \pi_x \in \Pi(G)\},$$
which depends on the choice of the functions $\pi_*, \pi_1, \ldots, \pi_g$.

\begin{prop}
    Let $g \in \mathbb N$. The set $\mathcal P$ as above is a polytope in $\mathbb R^g$, containing $0$ in its interior. One can recover the set $\Pi(G)$ (see Definition \ref{def:probability-hypergraph}) from the polytope $\mathcal P$:
    $$\Pi(G) = \left\{\pi_* + \sum_{x=1}^g a_x \pi_x \, :a \in \mathcal P\right\}.$$
    
    A different choice $\pi'_*, \pi'_1, \ldots, \pi'_g$ yields a polytope $\mathcal P'$ which is an affine transformation of $\mathcal P$.
\end{prop}
\begin{proof}
    A point $a \in \mathbb R^g$ is an element of $\mathcal P$ if and only if the following conditions are satisfied: 
    $$\forall v \in V, \qquad \pi_*(v) + \sum_{x=1}^g a_x \pi_x(v) \geq 0.$$
    Note that these conditions are affine in $a$ and define the facets of the polytope $\mathcal P$. The normalization condition Eq.~\eqref{eq:probability-hypergraph-normalization} is automatically satisfied:
    $$\forall e \in E \qquad \sum_{v \in e} \pi_*(v) + \sum_{x=1}^g \sum_{v \in e} a_x \pi_x(v) = \sum_{v \in e} \pi_*(v) = 1.$$
For the final claim, changing the base point $\pi_*$ to a different one amounts to a translation of the polytope, while a basis change for the vector space $\Pi^0(G)$ amounts to an invertible linear transformation of $\mathcal P$.
\end{proof}

It is remarkable that a big class of polytopes can be obtained in the way described above, a result due to Shultz \cite{shultz1974characterization}.

\begin{thm}[\cite{shultz1974characterization}]
    Any polytope $\mathcal P$ having vertices with rational coefficients can be obtained from a probability hypergraph $G$.
\end{thm}

\bigskip

We consider now a matrix version of the set $\Pi(G)$ from Definition \ref{def:probability-hypergraph}, by requiring that the function $\pi$ is matrix-valued. 

\begin{defi}
    Let $d \in \mathbb N$. Given $G$ a probability hypergraph, define, for all $d \geq 1$
    $$\Pi_d(G) := \left\{A : V \to \mathrm{PSD}_d \, : \, \forall e \in E, \, \sum_{v \in E} A(v) = I_d \right\}.$$
    We refer to the elements of $\Pi_d(G)$ as \emph{$G$-operators}.
\end{defi}

Clearly, $\Pi_d(G)$ is the set of POVMs corresponding to the hyperedges of $G$, where common vertices correspond to common effects. 

\bigskip

The set of semiclassical POVMs can be seen as coming from operators obtained from the extreme points of the polytope $\mathcal P$. Indeed, let $\sigma_1, \ldots, \sigma_k$ be the extreme points of $\Pi(G)$, and $C_1, \ldots, C_k$ be a POVM. The operators 
$$\forall v \in V \qquad A_v := \sum_{i=1}^k \sigma_i(v) C_i$$
are called semiclassical. In other words, operators $\{A_v\}_{v \in V}$ are called \emph{$G$-compatible} if there exists a POVM $\{C_\sigma\}$ indexed by the extreme points $\sigma \in \operatorname{ext} \Pi(G)$ such that
$$\forall v \in V \qquad A_v = \sum_{\sigma \in \operatorname{ext} \Pi(G)} \sigma(v) C_\sigma.$$

One can easily show the following result, connecting the formalism introduced in this section to the one from Section \ref{sec:polytope-compatibility}.

\begin{prop}\label{prop:P-G-equivalent}
Let $G$ be a probability hypergraph. To a $g$-tuple of operators $(B_x)_{x \in [g]}$ associate a tuple of self-adjoint operators $(A_v)_{v \in V}$  defined by
$$\forall v \in V, \qquad A_v := \pi_*(v)I_d + \sum_{x=1}^g \pi_x(v)B_x \in \mathcal M_d^{sa}(\mathbb C),$$ 
where $\pi_*, \pi_1, \ldots, \pi_g$ define the polytope $\mathcal P$ associated to $G$.
Then, we have the following equivalences: 
\begin{itemize}
    \item $(B_x)_{x \in [g]}$ are $\mathcal P$-operators $\iff$ $(A_v)_{v \in V}$ are $G$-operators; 
    \item $(B_x)_{x \in [g]}$ are $\mathcal P$-compatible $\iff$ $(A_v)_{v \in V}$ are $G$-compatible.
\end{itemize}
\end{prop}

We discuss next the relation between the notion of hypergraph compatibility introduced in this section, and the standard notion of compatibility for the quantum measurements associated to the hyperedges. In the spirit of \cite{guerini2018joint}, hypergraph compatibility is equivalent to standard compatibility plus the requirement that there exists a post-processing respecting the structure of the hypergraph. 

\begin{thm}\label{thm:G-symmetric-post-processing}
    Let $G=(V,E)$ be a probability hypergraph, and consider a tuple of $G$-operators $(A_v)_{v \in V} \subseteq \mathcal M_d^{sa}(\mathbb C)$, $d \in \mathbb N$. Consider also the POVMs $\hat A_{\cdot|e} = (A_v)_{v \in e}$ indexed by the hyperedges of $G$. The following assertions are equivalent:
    \begin{itemize}
        \item The tuple $(A_v)_{v \in V}$ is $G$-compatible.
        \item The measurements $\hat A_{\cdot|e}$ are compatible (in the standard sense of quantum mechanics), with the additional constraint that they can be post-processed from a single POVM $B= (B_\lambda)_{\lambda \in \Lambda}$
        $$\forall e \in E,  \forall v \in e, \qquad \hat A_{v|e} = A_v = \sum_{\lambda \in \Lambda} p(v|e, \lambda) B_\lambda$$
        using a post-processing $p$ respecting the symmetry of $G$:
        \begin{equation}\label{eq:G-symmetric-post-processing}
            \forall e,f \in E, \forall v \in e \cap f, \forall \lambda \in \Lambda, \qquad p(v|e, \lambda) = p(v|f, \lambda).
        \end{equation}
    \end{itemize}
\end{thm}
\begin{proof}
    The fact that the first point implies the second one is immediate: the tuple $A$ being $G$-compatible implies that there exists a POVM $C$ indexed by the extreme points of $\Pi(G)$ such that
    $$\forall v \in V \qquad A_v = \sum_{\sigma \in \operatorname{ext} \Pi(G)} \sigma(v) C_\sigma.$$
    Since the functions $\sigma$ have the property that 
    $$\forall e \in E, \qquad \sum_{v \in e} \sigma(v) =1,$$
    this yields the conclusion by setting $\Lambda = \operatorname{ext} \Pi(G)$ and $p(v|e, \sigma) = \sigma(v)$.

    For the reverse implication, the symmetry condition Eq.~\eqref{eq:G-symmetric-post-processing} implies that we can unambiguously define
    $$\forall v \in V, \forall \lambda \in \Lambda, \qquad \pi(v|\lambda) := p(v|e,\lambda) \quad \text{ for any } e \ni v.$$
    With this notation, we rewrite the (standard) compatibility of the POVMs $\hat A_{\cdot|e}$ as
    $$\forall v \in V, \qquad A_v = \sum_{\lambda \in \Lambda} \pi(v|\lambda) B_\lambda.$$
    Note that since $p$ is a post-processing, we have, for all $\lambda \in \Lambda$,  
    $$\forall e \in E \qquad \sum_{v \in e} \pi(v|\lambda) =1,$$
    hence $\pi(\cdot|\lambda) \in \Pi(G)$. We now decompose these elements in $\Pi(G)$ in terms of the extreme points:
    $$\forall \lambda \in \Lambda, \qquad \pi(\cdot|\lambda) = \sum_{\sigma \in \operatorname{ext} \Pi(G)} q(\sigma|\lambda) \sigma(\cdot),$$
    for some conditional probabilities $q(\cdot | \cdot)$. Putting everything together, we have 
    $$\forall v \in V, \qquad A_v = \sum_{\sigma \in \operatorname{ext} \Pi(G)} \sigma(v) \underbrace{\sum_{\lambda \in \Lambda} q(\sigma|\lambda) B_\lambda}_{C_\sigma}.$$
    It is easy to check that the $C_\sigma$ defined above form a POVM, concluding the proof.
\end{proof}

\bigskip

Let us now emphasize the considerations above with an example, which corresponds to the middle hypergraph in Figure \ref{fig:hypergraphs}, also reproduced below for convenience. This hypergraph $G$, having 5 vertices and 2 hyperedges, corresponds to two 3-outcome POVMs sharing one effect. The set of probability vectors on $G$ can be easily computed: 
$$\Pi(G) = \{(x,y,1-x-y,z,1-x-z) \, : \, 0 \leq x \leq 1 \text{ and } 0 \leq y \leq 1-x \text{ and } 0 \leq z \leq 1-x\}.$$
The allowed triples $(x,y,z)$ form a pyramid with a square base 
$$\operatorname{conv}\{((1,0,0), (0,0,0), (0,0,1), (0,1,0), (0,1,1)\},$$
depicted below in Figure \ref{fig:pyramid} (right panel). 

\begin{figure}
   \begin{center}
\begin{minipage}{6in}
  \centering
  \raisebox{-0.5\height}{\includegraphics{hperg-2.pdf}}
  \hspace*{2cm}
  \raisebox{-0.5\height}{\includegraphics[scale=0.4]{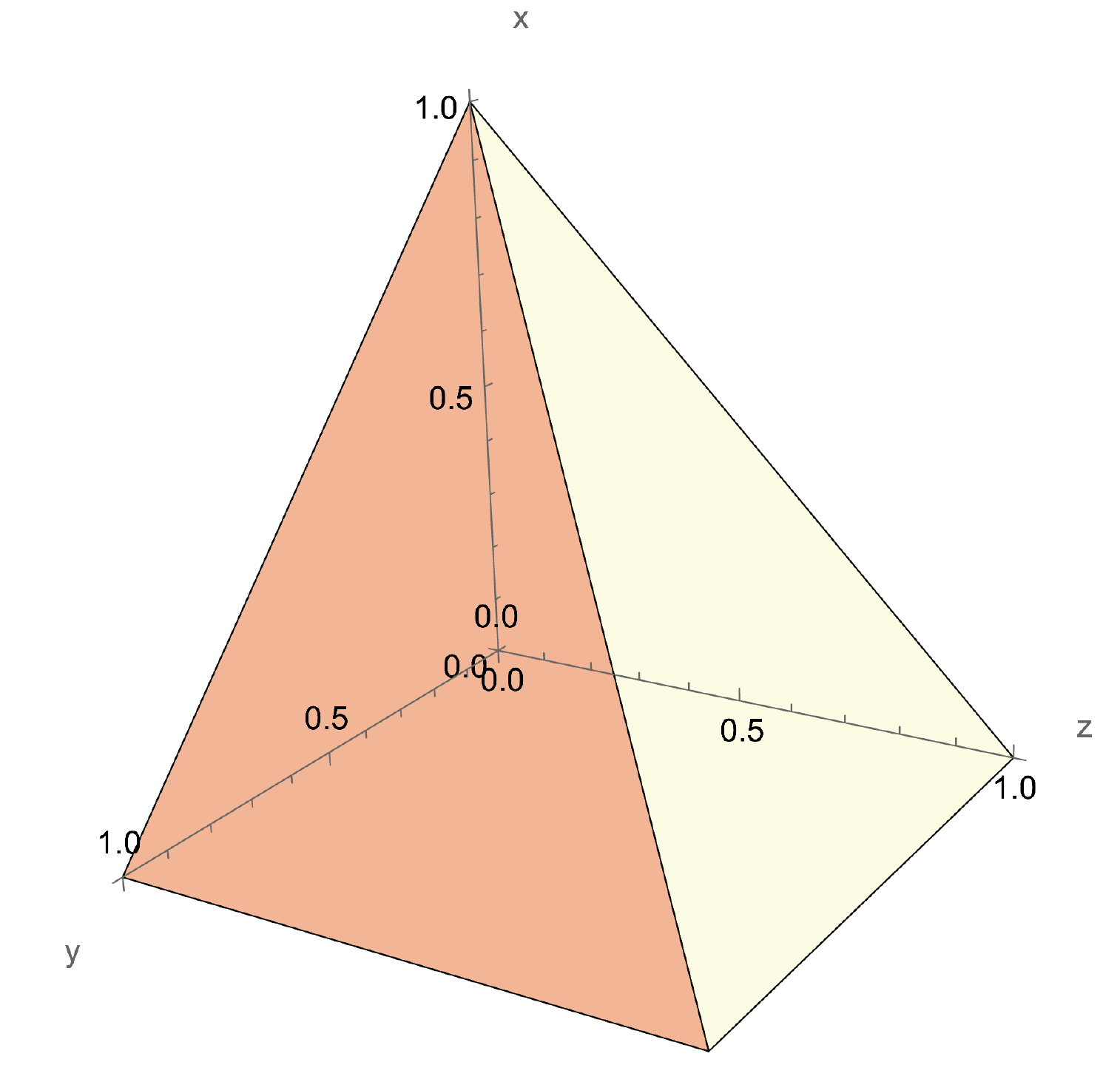}}
\end{minipage}
\end{center}
    \caption{A probability hypergraph and the associated pyramid.}
    \label{fig:pyramid}
\end{figure}

Considering the distinguished element $\pi_*:= (1,1,1,1,1)/3 \in \Pi(G)$ and the basis elements
\begin{align*}
    \pi_1 &= (1,0,-1,0,-1)\\
    \pi_2 &= (0,1,-1,0,0)\\
    \pi_3 &= (0,0,0,1,-1),   
\end{align*}
we obtain the associated polytope
\begin{equation}\label{eq:def-pyramid}
    \mathcal P = (-1/3, -1/3, -1/3) + \operatorname{conv}\{((1,0,0), (0,0,0), (0,0,1), (0,1,0), (0,1,1)\}.
\end{equation}

The following result is an application of Proposition \ref{prop:P-G-equivalent} and Theorem \ref{thm:G-symmetric-post-processing}
\begin{prop}\label{prop:3-outcomes-with-shared-effect}
Consider a triple of self-adjoint matrices $(A,B,C) \in (\mathcal M_d^{sa}(\mathbb C))^3$, $d \in \mathbb N$. Then
    \begin{itemize}
        \item $(A,B,C) \in 1/3(I,I,I) + \mathcal P_{\max}(d)$ if and only if $(A,B,I_d - A - B, C, I_d - A - C)$ are $G$-operators if and only if both triples $(A,B, I_d - A - B)$ and $(A,C, I_d - A - C)$ are POVMs. 
        \item $(A,B,C) \in 1/3(I,I,I) + \mathcal P_{\min}(d)$ if and only if $(A,B,I_d - A - B, C, I_d - A - C)$ are $G$-compatible if and only if the two POVMs $(A,B, I_d - A - B)$ and $(A,C, I_d - A - C)$ are compatible and the post-processing in Lemma \ref{lem:post-processing} satisfies $p_\lambda(1|1) = p_\lambda(1|2)$, where $p_\lambda(\cdot|1)$, $p_\lambda(\cdot|2)$ are conditional probabilities for all $\lambda$.
    \end{itemize}
\end{prop}

Note that the last point above corresponds to the existence of a joint POVM of the form
\begin{center}
\begin{tabular}{cccl}
\cline{1-3}
\multicolumn{1}{|c|}{$Q_1$} & \multicolumn{1}{c|}{0}    & \multicolumn{1}{c|}{0}    & $=A$                       \\ \cline{1-3}
\multicolumn{1}{|c|}{0}    & \multicolumn{1}{c|}{$Q_5$} & \multicolumn{1}{c|}{$Q_4$} & $=B$                       \\ \cline{1-3}
\multicolumn{1}{|c|}{0}    & \multicolumn{1}{c|}{$Q_3$} & \multicolumn{1}{c|}{$Q_2$} & $=I_d-A-B$                \\ \cline{1-3}
\phantom{$I_d$}$=A$\phantom{$-A-$}                         & \phantom{$I_d$}$=C$\phantom{$-A-$}                        & $=I_d-A-C$                 & $\sum=I_d$
\end{tabular}
\end{center}

One can ask a question similar to the one in Section \ref{sec:counterexample}. We present now an example of compatible POVMs which do not admit a joint POVM of the form above.  Again, $e_1$, $e_2$ are the standard basis vectors in $\mathbb C^2$ and $f_1 = 1/\sqrt{2}(e_1 + e_2)$, $f_2 = 1/\sqrt{2}(e_1 - e_2)$. Consider the following two 3-outcome qubit measurements: 
\begin{equation}\label{eq:example-123-145}
\left(\frac 1 2 I_2, \frac 1 2 e_1 e_1^\ast, \frac 1 2 e_2 e_2^\ast\right) \qquad \text{ and } \qquad \left(\frac 1 2 I_2, \frac 1 2 f_1 f_1^\ast, \frac 1 2 f_2 f_2^\ast\right).
\end{equation}

These POVM are compatible, as they are the marginals of the  following joint POVM: 
$$\frac 1 2 \, \cdot \, 
    \begin{tabular}{|c|c|c|}
        \hline
        0 & $ f_1 f_1^\ast$ & $ f_2 f_2^\ast$ \\ \hline
        $ e_1 e_1^\ast$ & 0 & 0 \\ \hline
        $ e_2 e_2^\ast$ & 0 & 0 \\ \hline
    \end{tabular}$$

However, when trying to write
        \begin{align*}
            \frac 1 2 I_2 &=Q_1\\
            \frac 1 2  e_1 e_1^\ast &=Q_4 + Q_5\\
            \frac 1 2  e_2 e_2^\ast &=Q_2 + Q_3\\
            \frac 1 2  f_1 f_1^\ast &=Q_3 + Q_5\\
            \frac 1 2  f_2 f_2^\ast &=Q_2 + Q_4,
        \end{align*}
we infer, for example, $Q_2 \sim e_2 e_2^\ast$ and $Q_2 \sim f_2 f_2^\ast$, hence $Q_2 = 0$. The same reasoning yields $Q_3=Q_4=Q_5 = 0$, implying $Q_1 = I_2$, which contradicts the first equality above. In conclusion, the POVMs from Eq.~\eqref{eq:example-123-145} are compatible but the corresponding effects do not belong to 
$\mathcal P_{\min}(2)$.

\section{Inclusion constants}\label{sec:inclusion-constants}

The main theme in this paper is understanding the relation between the sets of $\mathcal P$-operators and $\mathcal P$-compatible operators for a given polytope $\mathcal P$. Since the inclusion $\mathcal P_{\min}(d) \subseteq \mathcal P_{\max}(d)$ always holds, we would like to have a measure of how restrictive the compatibility condition is. This motivates the following definition, which is a specialization of Definition \ref{def:Delta-MCS} to matrix convex sets defined by polytopes. 

\begin{defi}
    Let $g$, $d \in \mathbb N$. For a polytope $\mathcal P \subseteq \mathbb R^g$ having $0$ in its interior, we define the set of \emph{inclusion constants} of $\mathcal P$ at level $d$ by
    $$\Delta_{\mathcal P}(d) := \{ s \in \mathbb R^g \, : \, s\cdot \mathcal P_{\max}(d) \subseteq \mathcal P_{\min}(d) \}.$$
    In words, $s \in \Delta_{\mathcal P}(d)$ if and only if for \emph{any} $g$-tuple of $\mathcal P$-operators $A \in \mathcal M_d^{\mathrm{sa}}(\mathbb C)^g$, the scaled tuple $(s_1 A_1, s_2 A_2, \ldots, s_g A_g)$ is $\mathcal P$-compatible.
\end{defi}

Computing the set of inclusion constants for a given polytope at a given level is in general a difficult problem, having many applications. For example, in the case of the hypercube, inclusion constants are related to the robustness of incompatibility of quantum effects \cite{bluhm2018joint}. More generally, the case of direct products of simplices is related to compatibility of quantum measurements \cite{bluhm2020compatibility}. The dual case, that of the cross polytope, is related to quantum steering \cite{bluhm2022maximal}.

Clearly, the set of inclusion constants of a polytope is convex. We can also show that it contains $0$ in its interior.
The following lemma is most likely known:
\begin{lem} \label{lem:0-in-int}
Let $\mathcal F$ be a matrix convex set consisting of $g$-tuples, $g \in \mathbb N$. If $0 \in \operatorname{int}\mathcal F_1$, then $0 \in \operatorname{int}\mathcal F_n$ for all $n \in \mathbb N$.
\end{lem}
\begin{proof}
    Let us consider vectors of the form $x e_i$, where $e_i$ is the $i$-th standard basis vector and $x \in \mathbb R$. As $0 \in \operatorname{int}\mathcal F_1$, there exists a $C > 0$ such that $x e_i \in \mathcal F_1$ for all $x$ with $|x|\leq C$ and all $i \in [g]$. As matrix convex sets are closed under direct sums, this implies that $(0, \ldots, 0,\operatorname{diag}[x_1, \ldots, x_n], 0, \ldots, 0) \in \mathcal F_n$ for all $x_1, \ldots, x_n$ such that $|x_j| \leq C$ for all $j \in [n]$, irrespective of the position of the diagonal matrix in the $g$-tuple. As the matrix convex set is closed under UCP maps, it is in particular closed under unitary conjugation, such that $(0, \ldots, 0, X, 0 \ldots, 0) \in \mathcal F_n$ for all $X \in \mathcal M_n^{\mathrm{sa}}(\mathbb C)$ with $\|X\|_\infty \leq C$, where the norm is the operator norm. As every level of the matrix convex set is convex, we infer that $(X_1, \ldots, X_g) \in \mathcal F_n$ if $X_j \in \mathcal M_n^{\mathrm{sa}}(\mathbb C)$ and $\|X_j\|_\infty \leq C/g$ for all $j \in [j]$. This proves the claim.
\end{proof}

\begin{cor}
    Let $\mathcal P$ be a polytope with $0$ in its interior. Then, $0 \in \operatorname{int} \Delta_{\mathcal P}(d)$ for all $d \in \mathbb N$.
\end{cor}
\begin{proof}
 Using Lemma \ref{lem:0-in-int}, it holds that $0\in \mathcal P_{\min}(d)$. Since $\mathcal P$ is a polytope, it is easy to see that $\mathcal P_{\max}(d)$ is bounded. Hence, there is a $C > 0$ such that $s \cdot \mathcal P_{\max}(d) \subseteq \mathcal P_{\min}(d)$ for all vectors $s$ with $\|s\|_2 \leq C$.
\end{proof}

The sets of inclusion constants are decreasing with the dimension: 
\begin{prop}\label{prop:inclusion-constants-dimension-decreasing}
    Let $\mathcal P\subseteq \mathbb R^g$ be a polytope containing $0$ in their interior, and $d \leq d'$, where $g$, $d$, $d' \in \mathbb N$. We have then 
    $$\Delta_{\mathcal P}(d) \supseteq \Delta_{\mathcal P}(d').$$
\end{prop}
\begin{proof}
    Let $s \in \Delta_{\mathcal P}(d')$ and $A \in \mathcal P_{\max}(d)$. We claim that $A \oplus 0_{d'-d} \in \mathcal P_{\max}(d')$. To see this, consider some fixed defining hyperplane $h_j$ and compute
    $$\sum_{x=1}^g h_j(x) (A_x \oplus 0_{d'-d})  = \left( \sum_{x=1}^g h_j(x) A_x \right) \oplus 0_{d'-d} \leq I_d \oplus 0_{d'-d} \leq I_{d'}.$$
    Alternatively, we could have used that matrix convex sets are closed under direct sums. From the fact that $s$ is an inclusion constant at level $d'$, it follows that $s \cdot (A \oplus 0_{d'-d}) \in \mathcal P_{\min}(d')$. Hence, there exist elements $v_1, \ldots, v_k \in \mathcal P$ and a POVM $C_1, \ldots, C_k$ of size $d'$ such that, for all $x \in [g]$, 
    $$s_x A_x \oplus 0_{d'-d} = \sum_{i=1}^k v_i(x) C_i.$$
    Defining $\tilde C_i$ to be the  $d \times d$ corner of $C_i$, which is still a POVM, shows that 
    $$s_x A_x = \sum_{i=1}^k v_i(x) \tilde C_i,$$
    finishing the proof.
\end{proof}

 The following result shows that the set of inclusion constants behaves nicely with respect to the Cartesian product and direct sum operations. Note that one can easily generalize the inclusions below to more than two polytopes.

\begin{prop}
    Let $\mathcal P\subseteq \mathbb R^g$ and $\mathcal Q\subseteq \mathbb R^n$ be two polytopes containing $0$ in their interior. Then, for all $d \geq 1$, $\Delta_{\mathcal P \times \mathcal Q}(d) \supseteq \Delta_{\mathcal P}(d) \oplus \Delta_{\mathcal Q}(d)$
\end{prop}
\begin{proof}
    Consider $(A,B) \in (\mathcal P \times \mathcal Q)_{\max}(d)$, and $s \in \Delta_{\mathcal P}(d)$. From Proposition \ref{prop:easy_consequences}, we have $A \in \mathcal P_{\max}(d)$, $B \in \mathcal Q_{\max}(d)$, and also $s\cdot A \in \mathcal P_{\min}(d)$. Since $0 \in \mathcal Q_{\min}(d)$, by using Proposition \ref{prop:easy_consequences} again, we have $(s \cdot A, 0) \in (\mathcal P \times \mathcal Q)_{\min}(d)$, proving that $(s,0) \in \Delta_{\mathcal P \times \mathcal Q}(d)$. A similar reasoning shows that 
    $$t \in \Delta_{\mathcal Q}(d) \implies (0, t) \in \Delta_{\mathcal P \times \mathcal Q}(d),$$
    concluding the proof of the claim.
\end{proof}

\begin{cor}\label{cor:polysimplex}
    Let $\mathcal P_1, \ldots, \mathcal P_n$ be polytopes containing 0 in their interior, $n \in \mathbb N$, $s_i \in \Delta_{\mathcal P_i}(d)$ for all $i \in [n]$, and $\lambda = (\lambda_1, \ldots, \lambda_n)$ a probability vector. Then 
    $$(\lambda_1 s_1, \lambda_2 s_2, \ldots, \lambda_n s_n) \in \Delta_{\times_{i=1}^n \mathcal P_i}(d).$$
    In particular, if the $\mathcal P_i$ are simplices, then, for all $d \geq 1$,
    $$n^{-1}(1,1, \ldots, 1) \in \Delta_{\times_{i=1}^n \mathcal P_i}(d).$$
\end{cor}
\begin{proof} 
The claim about the polysimplex follows from the fact that for any simplex $\mathcal P$, $\mathcal P_{\min} = \mathcal P_{\max}$ (see Proposition \ref{prop:min-max-simplex}). This fact corresponds to the fact that mixing quantum measurements with white noise using parameter $1/n$ yields compatible measurements \cite{Heinosaari2016}.
\end{proof}

We derive in the remaining part of this section some  general bounds for the set of inclusion constants of polytopes, applying them to the case of the Birkhoff body $\mathcal B_N$ (and the pyramid polytope discussed at the end of Section \ref{sec:POVM-common-elements} in the end of the paper). We shall discuss three different methods: 
\begin{itemize}
\item comparing with other polytopes
\item symmetrization
\item linear programming bounds
\end{itemize}

\subsection{Relating different polytopes} \label{sec:different-polytopes}

We start with the first method, where one relates the inclusion constants for two polytopes. 

\begin{prop}\label{prop:comparing-Deltas}
Let $g \in \mathbb N$. Give two polytopes $\mathcal P, \mathcal Q \subseteq \mathbb R^g$ containing $0$ in their interior, define the sets
\begin{align*}
\mathcal I_{\mathcal P \to  \mathcal Q} &:= \{s \in \mathbb R^g \, : \, s\cdot \mathcal P \subseteq \mathcal Q \}\\
\mathcal I_{\mathcal Q \to \mathcal P} &:= \{u \in \mathbb R^g \, : \, u\cdot \mathcal Q \subseteq \mathcal P \}.
\end{align*}
Then, for all $d \geq 1$, $d \in \mathbb N$,
$$\mathcal I_{\mathcal P \to  \mathcal Q} \cdot \Delta_{\mathcal Q}(d) \cdot \mathcal I_{\mathcal Q \to  \mathcal P} \subseteq \Delta_{\mathcal P}(d).$$
\end{prop}
\begin{proof}
Consider a tuple $A \in \mathcal M_d^{\mathrm{sa}}(\mathbb C)^g$ of $\mathcal P$ operators, and scaling vectors $s \in \mathcal I_{\mathcal P \to  \mathcal Q}$, $t \in \Delta_{\mathcal Q}(d)$, $u \in \mathcal I_{\mathcal Q \to \mathcal P}$. Since $A$ are $\mathcal P$-operators, we have 
$$\forall \rho \in \mathcal M_d^{1,+}(\mathbb C) \quad (\langle A_x, \rho \rangle)_{x \in [g]} \in \mathcal P \implies (\langle s_x A_x, \rho \rangle)_{x \in [g]} \in \mathcal Q,$$
that is $s \cdot A \in \mathcal Q_{\max}(d)$. In turn, using the fact that $t$ is an inclusion constant for $\mathcal Q$, we have that $t \cdot s \cdot A \in \mathcal Q_{\min}(d)$. In particular, this means that there exists a POVM $C_1, \ldots, C_k$ such that 
$$t \cdot s \cdot A = \sum_{i=1}^k w_i \otimes C_i,$$
where $w_1, \ldots, w_k \in \mathbb R^g$ are elements (e.g.~the extreme points) of $Q$. Scaling by $u$ gives
$$u \cdot t \cdot s \cdot A = \sum_{i=1}^k (u\cdot w_i) \otimes C_i,$$
with $u \cdot w_i \in \mathcal P$ now, proving the claim.
\end{proof}

We shall now apply the result above to the case of the Birkhoff body $\mathcal B_N$ and a \emph{polysimplex}, i.e.~a Cartesian product of simplices \cite[Section III]{jencova2018incompatible}, \cite[Section 5]{Bluhm2022GPT}. Note that this method has been used previously to derive inclusion constants for free spectrahedra by comparing them to more symmetric matrix convex sets for which the set of inclusion constants was better known \cite[Section 7]{bluhm2020compatibility}. Starting from this point, we are focusing on ``flat'' inclusion constants of the form $s(1,1,\ldots, 1)$, so we are going to simply identify the scalar $s$ with the corresponding flat vector.

Consider the polysimplex 
$$\mathcal Q:=\bigtimes_{i=1}^{N-1} \mathcal P_N \subseteq \mathbb R^{(N-1)^2} \cong \mathcal M_{N-1}(\mathbb R),$$
where $\mathcal P_N \subseteq R^{N-1}$ is the simplex defined in Eq.~\eqref{eq:def-simplex}. Notice that $\mathcal B_N \subseteq \mathcal Q$, since the extreme points of $\mathcal B_N$ are a subset of those of $\mathcal Q$; this shows that $1 \in \mathcal I_{\mathcal B_N \to \mathcal Q}$. In order to find the largest scalar $u \geq 0$ such that $u \in \mathcal I_{\mathcal Q \to \mathcal B_N}$, note that the extreme points of $\mathcal Q$ are of the form
$$q_f := -\frac{1}{N} J_{N-1} + \begin{bmatrix} e_{f(1)} \\ e_{f(2)} \\ \vdots \\ e_{f(N-1)} \end{bmatrix},$$
where $f:[N-1] \to [N]$ is an arbitrary function, and where we set $e_N = 0$. We have to determine by how much we have to scale these extreme points in order for them to satisfy the hyperplane equations Eqs.~\eqref{eq:BN-face-elt}-\eqref{eq:BN-face-sum}. Eqs.~\eqref{eq:BN-face-elt} and \eqref{eq:BN-face-row} are satisfied without scaling ($u \leq 1$). For Eq.~\eqref{eq:BN-face-col}, fix $j \in [N-1]$ and compute
$$\sum_{i=1}^{N-1} q_f(i,j) = -\frac{N-1}{N} + \left|\{i \in [N-1] \, : \, f(i) = j\}\right| \leq -\frac{N-1}{N} + N-1 = \frac{(N-1)^2}{N}.$$
Hence, for $u \cdot q_f \in \mathcal B_N$ to hold, we need, in the worst case, $u \leq 1/(N-1)^2$. We leave to the reader the case of the condition Eq.~\eqref{eq:BN-face-sum}, which yields the same inequality, $u \leq 1/(N-1)^2$. We conclude that 
$$\frac{1}{(N-1)^2} \in \mathcal I_{\mathcal Q \to \mathcal B_N}.$$
From Corollary \ref{cor:polysimplex} we have, for all $d \geq 1$,  $1/(N-1) \in \Delta_{\mathcal Q}(d)$. Putting all these facts into Proposition \ref{prop:comparing-Deltas}, we arrive at the following result. 

\begin{cor}\label{cor:BN-inclusion-conparison}
For any $N \geq 2$ and $d \geq 1$, 
$$\frac{1}{(N-1)^3} \in \Delta_{\mathcal B_N}(d).$$
\end{cor}

\subsection{Symmetrization} \label{sec:symmetrization}

In this section we show that matrix convex sets with \emph{symmetric} first level admit a non-trivial, dimension-dependent, inclusion constant vector. This vector does not depend on the length of the tuple. We start with a general result regarding matrix convex sets, specialize to polytope inclusion constants, and then derive constants for the Birkhoff body $\mathcal B_N$.

The following result is inspired by \cite[Proposition 8.1]{helton2019dilations} for free spectrahedra, which are a special class of matrix convex sets (see also \cite[Proposition 7.2]{bluhm2018joint}).
\begin{prop}\label{prop:symmetric-inclusion}
     Let $d \in \mathbb N$ and let $\mathcal F, \mathcal G$ be matrix convex sets such that $\pm \mathcal F(1) \subseteq \mathcal G(1)$. Then,
    \begin{align*}
        \frac{1}{2d-1} \mathcal F(d) \subseteq  \mathcal G(d) & \mathrm{\quad if~} d \mathrm{~is~even}, \\
        \frac{1}{2d+1} \mathcal F(d) \subseteq  \mathcal G(d) & \mathrm{\quad if~} d \mathrm{~is~odd}.
    \end{align*} 
\end{prop}

\begin{proof}
    Let $e_i$, $i \in [d]$, be the standard basis. We define for $s \neq t$
    \begin{align*}
        e^{\pm}_{s,t} &= \frac{1}{\sqrt{2}} (e_s \pm e_t),\\
        \phi^{\pm}_{s,t} &= \frac{1}{\sqrt{2}} (e_s \pm i e_t).
    \end{align*}
    It is easy to see that $X \mapsto (e^{\pm}_{s,t})^\ast X e^{\pm}_{s,t}$, $X \mapsto e_{t}^\ast X e_t$ and $X \mapsto (\phi^{\pm}_{s,t})^\ast X \phi^{\pm}_{s,t}$ are UCP maps from $\mathcal M_d(\mathbb C)$ to $\mathbb C$. 

Given $(Y_1, \ldots, Y_g) \in \mathcal F(d)$, it follows that, for $s \neq t \in [d]$,
    \begin{align*}
       \pm \left[e_s^\ast Y_i e_s\right]_{i \in [g]} &\in \pm \mathcal F(1) \subseteq \mathcal G(1)\\    
       \pm \left[(e^{+}_{s,t})^\ast Y_i e^{+}_{s,t}\right]_{i \in [g]} &\in \pm \mathcal F(1) \subseteq \mathcal G(1)\\
       \pm \left[(e^{-}_{s,t})^\ast Y_i e^{-}_{s,t}\right]_{i \in [g]} &\in \pm \mathcal F(1) \subseteq \mathcal G(1)\\
       \pm \left[(\phi^{+}_{s,t})^\ast Y_i \phi^{+}_{s,t}\right]_{i \in [g]} &\in \pm \mathcal F(1) \subseteq \mathcal G(1)\\
       \pm \left[(\phi^{-}_{s,t})^\ast Y_i \phi^{-}_{s,t}\right]_{i \in [g]} &\in \pm \mathcal F(1) \subseteq \mathcal G(1).  
   \end{align*}
   Note that we can recover the real and imaginary parts of the entries of an arbitrary matrix $X$ from the quantities above:
   $$X_{ss} = e_s^* X e_s, \quad \operatorname{Re}X_{st} = \frac 1 2 \left[ (e^{+}_{s,t})^\ast X e^{+}_{s,t} - (e^{-}_{s,t})^\ast X e^{-}_{s,t}\right], \quad \operatorname{Im}X_{st} = \frac 1 2 \left[ (\phi^{-}_{s,t})^\ast X \phi^{-}_{s,t} - (\phi^{+}_{s,t})^\ast X \phi^{+}_{s,t}\right].$$
   We conclude that, for all $s$, $t \in [d]$,    
\begin{align*}
    &\pm(\operatorname{Re}(Y_1)_{s,t}, \ldots, \operatorname{Re}(Y_g)_{s,t}) \in \mathcal G(1), \\
    &\pm(\operatorname{Im}(Y_1)_{s,t}, \ldots, \operatorname{Im}(Y_g)_{s,t}) \in \mathcal G(1).
\end{align*}
Let us consider the set 
\begin{equation*}
    \mathcal I := \{(s,t): s, t \in [d], s<t\}.
\end{equation*}
Without loss of generality, we can assume that the dimension $d$ is even by going to $d+1$ if necessary and using Proposition \ref{prop:inclusion-constants-dimension-decreasing}. Then, we can partition $\mathcal I$ into $d-1$ subsets $\mathcal J_k$ of $d/2$ tuples each such that $\cup_{(s,t) \in \mathcal J_k}\{s,t\} = [d]$ for all $k \in [d-1]$, i.e., no index appears twice. This partitioning is possible, since it is equivalent to an edge-coloring of the complete graph $K_d$ such that the edges of each color form a perfect matching. The edges of each color $k$ then correspond to $\mathcal J_k$. Such a coloring exists for even $d$, e.g., by Baranyai's theorem (see, e.g., \cite{van2001course}). 

Let us fix $k \in [d-1]$ for now and consider $\mathcal J = \{(s_1, t_1), \ldots, (s_{d/2}, t_{d/2})\}$. As matrix convex sets are closed under direct sums, it follows that
\begin{align}
    (\operatorname{Re}(Y_i)_{s_1,t_1}\oplus \ldots \oplus \operatorname{Re}(Y_i)_{s_{d/2},t_{d/2}} \oplus-\operatorname{Re}(Y_i)_{s_1,t_1}\oplus \ldots \oplus - \operatorname{Re}(Y_i)_{s_{d/2},t_{d/2}})_{i \in [g]} &\in \mathcal G(d), \label{eq:real-vector}\\
    (-\operatorname{Im}(Y_i)_{s_1,t_1}\oplus \ldots \oplus -\operatorname{Im}(Y_i)_{s_{d/2},t_{d/2}} \oplus\operatorname{Im}(Y_i)_{s_1,t_1}\oplus \ldots \oplus  \operatorname{Im}(Y_i)_{s_{d/2},t_{d/2}})_{i \in [g]} &\in \mathcal G(d) \label{eq:imaginary-vector}\\
    ((Y_i)_{1,1}\oplus \ldots \oplus (Y_i)_{d,d})_{i \in [g]} &\in \mathcal G(d). \nonumber
\end{align}
Furthermore, matrix convex sets are closed under unitary conjugation. Let us define a unitary $U$ that maps 
\begin{align*}
    e_j \mapsto e^+_{s_j, t_j} \quad \forall j \in [d/2] \\
     e_{d/2+j} \mapsto e^-_{s_j, t_j} \quad \forall j \in [d/2] 
\end{align*}
and a unitary $V$ that maps 
\begin{align*}
    e_j \mapsto \phi^+_{s_j, t_j} \quad \forall j \in [d/2] \\
     e_{d/2+j} \mapsto \phi^-_{s_j, t_j} \quad \forall j \in [d/2]
\end{align*}
Thus, conjugating Eq.~\eqref{eq:real-vector} by $U$ and Eq.~\eqref{eq:imaginary-vector} by $V$, we obtain that 
\begin{align*}
    &\left(\sum_{j \in [d/2]}\operatorname{Re}(Y_\ell)_{s_j,t_j} (e_{s_j} e_{t_j}^\ast + e_{t_j}e_{s_j}^\ast)\right)_{\ell \in [g]} \in \mathcal G(d), \\
    &\left(\sum_{j \in [d/2]}i\operatorname{Im}(Y_\ell)_{s_j,t_j} (e_{s_j} e_{t_j}^\ast - e_{t_j}e_{s_j}^\ast)\right)_{\ell \in [g]} \in \mathcal G(d),
\end{align*}
because
$e^{+}_{s,t}(e^{+}_{s,t})^* - e^{-}_{s,t}(e^{-}_{s,t})^* = e_s e_t^\ast + e_t e_s^\ast$ and $\phi^{+}_{s,t}(\phi^{+}_{s,t})^* - \phi^{-}_{s,t}(\phi^{-}_{s,t})^* = i(-e_s e_t^\ast + e_t e_s^\ast)$. This construction can be repeated for all $\mathcal J_k$, $k \in [d-1]$.
Taking uniform convex combinations of these elements, we infer that for even $d$
\begin{equation*}
   \frac{1}{2d-1} (Y_1, \ldots, Y_g) \in \mathcal G(d).
\end{equation*}
As $(Y_1, \ldots, Y_g) \in \mathcal F(d)$ was arbitrary, the assertion follows. 
\end{proof}

Note that the inclusion constant obtained for free spectrahedra in \cite{bluhm2018joint} is $2d$, which is slightly worse than our result for even $d$. Remark 7.5 in \cite{bluhm2018joint} indicates that our result for even $d$ is optimal.

\begin{cor}\label{cor:inclusion-symmetric}
    Let $\mathcal P \in \mathbb R^g$ be a \emph{symmetric} polytope, i.e.~$\mathcal P = -\mathcal P$. Then, for all $d \geq 2$, 
    $$\frac{1}{\delta}\Big(\underbrace{1, 1, \ldots, 1}_{g \text{ times}} \Big) \in \Delta_{\mathcal P}(d),$$
    where
    \begin{equation}\label{eq:def-delta}
        \delta = \begin{cases}
        2d-1 &\qquad \text{ if $d$ is even}\\
        2d+1 &\qquad \text{ if $d$ is odd}.
    \end{cases}
    \end{equation}
    
\end{cor}

In order to apply the corollary above to the Birkhoff body $\mathcal B_N$, we first have to symmetrize it: we seek the best (i.e.~largest) constant $s \in [0,1]$ such that 
\begin{equation}\label{eq:s-BN}
    s (-\mathcal B_N) \subseteq \mathcal B_N \iff s \underbrace{\operatorname{conv}(-\mathcal B_N \cup \mathcal B_N)}_{=: \mathcal P} \subseteq \mathcal B_N.
\end{equation}

In this way, we have, using Propositions  \ref{prop:symmetric-inclusion} and \ref{prop:polytope-inclusion}: 
$$X \in (\mathcal B_N)_{\max}(d) \implies X \in \mathcal P_{\max}(d) \implies  \frac 1 \delta \cdot X \in \mathcal P_{\min}(d) \implies \frac s \delta \cdot X \in (\mathcal B_N)_{\min}(d).$$

Let us now compute the best constant $s$ in Eq.~\eqref{eq:s-BN}. To do this, we have to find the largest value the facet inequalities from Proposition \ref{prop:birkhoff-body} attain on the negative of the extremal points of $\mathcal B_N$. The facets of $\mathcal B_N$ from Eqs.~\eqref{eq:BN-face-elt}-\eqref{eq:BN-face-sum} correspond to the matrices
$$h_{ij} = \begin{cases}
    -N e_i e_j^* &\qquad \text{ if } i,j\in [N-1]\\
    N \sum_{k=1}^{N-1} e_i e_k^* &\qquad \text{ if } i\in [N-1], \, j=N\\
    N \sum_{k=1}^{N-1} e_k e_j^* &\qquad \text{ if } j\in [N-1], \, i = N\\
    -N \sum_{i,j=1}^{N-1} e_i e_j^* &\qquad \text{ if } i,j = N.
\end{cases}$$
It is easy to see that the maximum value of the quantities
$$\left\langle h_{ij} , -\left(P_\sigma^{(N-1)} - \frac J N\right) \right \rangle$$
is $N-1$, for all $i,j \in [N]$ and $\sigma$ a permutation of $N$ elements. Hence, the largest $s$ for which Eq.~\eqref{eq:s-BN} holds is $s=1/(N-1)$. We have thus the following corollary regarding the Birkhoff body. 

\begin{cor}\label{cor:BN-inclusion-symmetrization}
For any $N \geq 2$ and $d \geq 1$, 
$$\frac{1}{(N-1)\delta} \in \Delta_{\mathcal B_N}(d),$$
where $\delta$ is the dimension dependent constant from Eq.~\eqref{eq:def-delta}.
\end{cor}

\subsection{Inclusion constants from linear programming} \label{sec:linear-programming}

Recall that a polytope $\mathcal P$ has two equivalent representations: one in terms of vertices (the ``V'' representation) and one in terms of supporting hyperplanes (the ``H'' representation). To the $k$ extreme points $v_1, \ldots, v_k \in \mathbb R^g$ of a polytope $\mathcal P$ we associate the matrix 
$$V := \sum_{i=1}^k v_i e_i^* \in \mathcal M_{g \times k}(\mathbb R)$$
having the $v_i$ as columns. Similarly, if $\{x \, : \, \langle h_j, x \rangle \leq 1\}_{j \in [r]}$ are the halfspaces defining $\mathcal P$ (recall that $\mathcal P$ contains 0 in its interior), we introduce the matrix 
$$H := \sum_{j=1}^r e_j^* h_j \in \mathcal M_{r \times g}(\mathbb R)$$
having the $h_j$ as columns. We extend these matrices by appending ones (we denote by $1_n \in \mathbb R^n$ the all-1 vector): 
$$\hat V := \begin{bmatrix} V \\ 1_k^\top \end{bmatrix} \in \mathcal M_{(g+1) \times k}(\mathbb R)$$
and 
$$\hat H := \begin{bmatrix} -H & 1_r\end{bmatrix} \in \mathcal M_{r \times (g+1)}(\mathbb R).$$
The matrices $\hat V$ and $\hat H$ of the polytope $\mathcal P$ are associated to the \emph{slack matrix} of $\mathcal P$ \cite{yannakakis1991expressing}: $S_{\mathcal P} = \hat H \hat V$.

In the same vein, to a $g$-tuple $A$ of self-adjoint operators, we associate the $(g+1)$-tuple $$\hat A := \begin{bmatrix} A \\ I_d \end{bmatrix} \in \mathbb R^{g+1} \otimes \mathcal M^{sa}_d(\mathbb C).$$

\begin{lem}\label{lem:P-operators-hat-H}
Let $g$, $r\in \mathbb N$. A $g$-tuple of self-adjoint operators $A$ are $\mathcal P$-operators iff $\hat H \hat A$ is entrywise positive semidefinite. In particular, for a vector $x \in \mathbb R^g$, $x \in \mathcal P \iff \hat H \begin{bsmallmatrix}x\\1\end{bsmallmatrix} \in \mathbb R_+^r$.
\end{lem}
\begin{proof}
    Requiring that the $r$ blocks of $\hat H \hat A$ are positive semidefinite is equivalent to: 
    $$\forall j \in [r] \qquad \sum_{x=1}^g \hat H_{j,x}A_x + 1 \cdot I_d = -\sum_{x=1}^g h_j(x)A_x + I_d \geq 0,$$
which is precisely the condition that $A \in \mathcal P_{\max}(d)$ from Proposition \ref{prop:P-hj}.    
    \end{proof}

\begin{lem}\label{lem:P-compatible-hat-V}
Let $g$, $k\in \mathbb N$. A $g$-tuple of self-adjoint operators $A$ are $P$-compatible iff there exists an entrywise positive $k$-tuple $C$ such that  $\hat V C = \hat A$. 
\end{lem}
\begin{proof}
    The equation for the extended vectors is equivalent to 
    $$VC = A \quad \text{ and } \quad 1_k^\top C = I_d.$$
    While the latter equation is equivalent to the normalization condition $\sum_{i=1}^k C_i = I_d$, the former is equivalent to 
    $$ \sum_{i=1}^k v_i  \otimes C_i  = \sum_{x=1}^g  e_x \otimes A_x,$$
    where we recall that $v_1, \ldots, v_k \in \mathbb R^g$ are the extreme points of $\mathcal P$. This, in turn, gives
    $$\forall x \in [g] \qquad A_x = \sum_{i=1}^k v_i(x) C_i$$
    which is precisely the condition that $A \in \mathcal P_{\min}(d)$ from Definition \ref{defi:P-operators-P-compatible}.
\end{proof}

Note that polytopes $P \subseteq \mathbb R^g$ containing zero in their interior have at least $g+1$ extreme points.

\begin{thm}\label{thm:bound-from-LP}
    Let $d$, $g$, $k$, $r \in \mathbb N$. Given $s \in \mathbb R^g$, if there exists an entrywise non-negative matrix $T \in \mathcal M_{k \times r}(\mathbb R_+)$ such that 
    $$\operatorname{diag}(s_1, s_2, \ldots, s_g, 1)=:\hat D_s =\hat V T \hat H ,$$
    then $s \in \Delta_{\mathcal P}(d)$, for all $d \geq 1$.
\end{thm}
\begin{proof}
    Consider a scaling vector $s$ satisfying the hypotheses of the statement, and let $A \in \mathbb R^g \otimes \mathcal M^{sa}_d(\mathbb C)$ a $g$-tuple of $P$-operators. From the hypothesis, we have 
    $$\hat V T \hat H = \hat D_s \implies \hat V \underbrace{T \hat H \hat A}_{C} = \hat D_s \hat A = \widehat{s \cdot A}.$$
    Since $A$ are $\mathcal P$-operators, it follows from Lemma \ref{lem:P-operators-hat-H} that $\hat H \hat A$ is entrywise positive semidefinite, hence so is $C$. We conclude that $s \cdot A$ are $\mathcal P$-compatible by applying Lemma \ref{lem:P-compatible-hat-V}.
\end{proof}
Note that the existence of the non-negative matrix $T$ can be formulated as a \emph{linear program}, so the theorem above provides a computationally tractable way to produce elements of the inclusion constants set
$$\bigcap_{d \geq 1} \Delta_{\mathcal P}(d).$$

Let us now apply this to the Birkhoff body $\mathcal B_N$. Recall that $\mathcal B_N$ has $k=N!$ extreme points and $r=N^2$ facets.

\begin{prop}\label{prop:LP-birkhoff}
Let $s$ be the constant vector with entries $1/(N-1)$ and $T = [\mathds{1}_{\langle h_j, v_i \rangle \neq 1} / (N \cdot N!)]_{ij}$. Then $\hat V T \hat H = \hat D_s$, proving that $1/(N-1) \in \Delta_{\mathcal B_N}(d)$ for all $d \geq 1$.
\end{prop}

\begin{proof}
    The proof follows from direct calculation. First, note that the condition $\hat V T \hat H = \hat D_s$ is equivalent to 
    $$ -V T  H = \operatorname{diag}(s)=:D_s, \quad 1_k^\top TH = 0, \quad VT 1_r = 0, \quad \langle 1_k, T1_r \rangle = 1.$$
    In other words, $T$ is a bistochastic matrix, with marginals in the kernels of $V$, resp.~$H^\top$, which satisfies $-VTH  = D_s$. The columns of $V$ are the extreme points of $\mathcal B_N$, and correspond to shifted truncated permutation matrices. They have coordinates
    $$\forall \pi \in \mathcal S_N, \, \forall x,y \in [N-1], \qquad v_\pi(x,y) = -\frac{1}{N} + \mathds 1_{\pi(x) = y}.$$

    The rows of $H$ correspond to the inequalities \eqref{eq:BN-face-elt}-\eqref{eq:BN-face-sum} defining the Birkhoff body, and they can be indexed by $[N]^2$; they have coordinates
    \begin{align*}
    \forall i,j \in [N-1] &\qquad h_{ij}(x,y) = -N \mathds 1_{i=x, j=y}\\
    j=N, \, \forall i \in [N-1] &\qquad h_{ij}(x,y) = N \mathds 1_{i=x}\\
    i=N, \, \forall j \in [N-1] &\qquad h_{ij}(x,y) = N \mathds 1_{j=y}\\
    i,j =N &\qquad h_{ij}(x,y) = -N .
    \end{align*}

    In particular, we have the crucial relation 
    $$\langle h_{ij}, v_\pi \rangle = \begin{cases}
        1 &\qquad \text{ if } \pi(i) \neq j\\
        1-N &\qquad \text{ if } \pi(i) = j.\end{cases}$$
        
    It follows that a given extreme point $v_\pi$ belongs to all the facets $\{x \in \mathbb R^g \, : \, \langle h_{ij},x \rangle=1\}$ such that $\pi(i) \neq j$; there are $N^2-N$ such facets. In other words, the $\pi$-th row of the matrix $T$ has precisely $N$ non-zero entries. It follows that the row sums of $T$ are all equal to $1/N!$. The condition 
    $$VT1_r = \frac{1}{N!}V 1_{N^2} = 0$$
    follows from the fact that the average of the extreme points $v_\pi$ is $0$; equivalently, this can be seen to follow from the fact that the average of the permutation matrices is the matrix $J/N$. A similar reasoning yields the condition $1_{N!}^\top TH = 0$. Finally, the main condition $-VTH=I/(N-1)$ follows from a simple (but tedious) analysis of the expression:
    $$(VTH)_{(x_1, y_1), (x_2, y_2)} = \sum_{\pi \in \mathcal S_N} \sum_{i,j=1}^N \frac{\mathds 1_{\pi(i)=j}}{N\cdot N!} v_{\pi}(x_1,y_1) h_{ij}(x_2,y_2).$$
\end{proof}
\begin{remark}
    Note how the Birkhoff body inclusion constant $1/(N-1)$ derived in this section is better than the ones obtained using the comparison technique (Corollary \ref{cor:BN-inclusion-conparison}) and the symmetrization technique (Corollary \ref{cor:BN-inclusion-symmetrization}). 
\end{remark}
\begin{cor}\label{cor:BN-semiclassical}
    Let $N$, $d \in \mathbb N$. For any magic square $A \in \mathcal M_N(\mathcal M_d^{\textrm{sa}}(\mathbb C))$, the convex combination
    \begin{equation}\label{eq:BN-convex}
        B:=\frac{1}{N-1} A + \frac{N-2}{N-1} \cdot  \frac{J_N}{N}
    \end{equation}
    is a semiclassical magic square. 
\end{cor}
\begin{proof}
    This follows from Proposition \ref{prop:LP-birkhoff} above and from the identity
    $${B}^{(N-1)} - \frac{J_{N-1}}{N} = \frac{1}{N-1}\left( {A}^{(N-1)} - \frac{J_{N-1}}{N} \right).$$
\end{proof}

In \cite[Theorem 12-(ii)]{cuevas2020quantum}, the authors have shown that a magic square $A$ having the property that
\begin{equation}\label{eq:condition-magic-square-semiclassical}
    \forall \pi \in \mathcal S_N \qquad \sum_{k=1}^N A_{k,\pi(k)} \geq \frac{N-2}{N-1}I_d
\end{equation}
is semiclassical. This fact can obtained from the following rewriting of $A$:
\begin{align*}
    A &= \sum_{i,j=1}^N e_ie_j^* \otimes A_{ij} \\
    &= \sum_{\pi \in S_N} \pi \otimes \frac{1}{(N-2)!N}\left[ \sum_{k=1}^N A_{k,\pi(k)} - \frac{N-2}{N-1}I_d\right].
\end{align*}
Indeed, if the condition Eq.~\eqref{eq:condition-magic-square-semiclassical} is satisfied, then one can take 
$$C_\pi := \frac{1}{(N-2)!N}\left[ \sum_{k=1}^N A_{k,\pi(k)} - \frac{N-2}{N-1}I_d\right] \geq 0$$
as the POVM certifying semiclassicality in Definition \ref{def:magic-square}. We would like to end this section by emphasizing the close relation between this condition and the one from Corollary \ref{cor:BN-semiclassical}. A magic square $B$ can we written as in Eq.~\eqref{eq:BN-convex} if and only if $B_{ij} \geq (N-2)/(N(N-1))I_d$ for all $i,j \in [N]$. This condition implies the one in Eq.~\eqref{eq:condition-magic-square-semiclassical}, but the converse is not true. Indeed, consider the bistochastic matrix 
$$\begin{bmatrix}
     1 & 0 & 0 \\
    0 & \frac 1 2 & \frac 1 2 \\
    0 & \frac 1 2 & \frac 1 2
\end{bmatrix}$$
which clearly satisfies Eq.~\eqref{eq:condition-magic-square-semiclassical} but which cannot be written as in Eq.~\eqref{eq:BN-convex} because of its $0$ entries.

\bigskip 

Finally, let us apply Theorem \ref{thm:bound-from-LP} to the case of the pyramid polytope from Eq.~\eqref{eq:def-pyramid} corresponding to two POVMs with three outcomes sharing one effect. Recall from Section \ref{sec:POVM-common-elements} that this polytope $\mathcal P$ is a pyramid with a square basis, having defining matrices $\hat V$ and $\hat H$ given respectively by
$$\hat V = \begin{bmatrix}
    2/3 & -1/3 & -1/3 & -1/3 & -1/3 \\
    -1/3 & -1/3 & -1/3 & 2/3 & 2/3 \\
    -1/3 & -1/3 & 2/3 & -1/3 & 2/3 \\
    1 & 1 & 1 & 1 & 1
\end{bmatrix} \quad \text{ and } \quad \hat H = \begin{bmatrix}
    3 & 0 & 0 & 1 \\
    0 & 3 & 0 & 1 \\
    0 & 0 & 3 & 1 \\    
    -3 & -3 & 0 & 1 \\    
    -3 & 0 & -3 & 1
\end{bmatrix}.$$

A simple calculation shows that taking 
$$T:=\frac{1}{30}\begin{bmatrix}
    6 & 1 & 1 & 1 & 1 \\
    1 & 0 & 0 & 2 & 2 \\
    1 & 0 & 2 & 2 & 0 \\
    1 & 2 & 0 & 0 & 2 \\
    1 & 2 & 2 & 0 & 0 
\end{bmatrix}$$
yields $\hat V T \hat H = \operatorname{diag}(2/5, 2/5, 2/5, 1)$, showing that $(2/5, 2/5, 2/5) \in \Delta_{\mathcal P}(d)$ for all dimensions $d$. This implies that the two POVMs $(\frac{2}{5} A + \frac{1}{5} I_d, \frac{2}{5} B + \frac{1}{5} I_d, \frac{3}{5}I_d - \frac{2}{5}A - \frac{2}{5}B)$ and $(\frac{2}{5} A + \frac{1}{5} I_d, \frac{2}{5} C + \frac{1}{5} I_d, \frac{3}{5}I_d - \frac{2}{5}A - \frac{2}{5}C)$ are compatible. In addition, the proof of Theorem \ref{thm:bound-from-LP} shows that the joint measurement has the form 
\begin{center}
\begin{tabular}{cccl}
\cline{1-3}
\multicolumn{1}{|c|}{$Q_1$} & \multicolumn{1}{c|}{0}    & \multicolumn{1}{c|}{0}    & $=\frac{2}{5} A + \frac{1}{5} I_d$                       \\ \cline{1-3}
\multicolumn{1}{|c|}{0}    & \multicolumn{1}{c|}{$Q_5$} & \multicolumn{1}{c|}{$Q_4$} & $=\frac{2}{5} B + \frac{1}{5} I_d$                       \\ \cline{1-3}
\multicolumn{1}{|c|}{0}    & \multicolumn{1}{c|}{$Q_3$} & \multicolumn{1}{c|}{$Q_2$} & $=\frac{3}{5}I_d - \frac{2}{5}A - \frac{2}{5}B$                \\ \cline{1-3}
\phantom{$I_d$}$=\frac{2}{5} A + \frac{1}{5} I_d$\phantom{$-A-$}                         & \phantom{$I_d$}$=\frac{2}{5} C + \frac{1}{5} I_d$\phantom{$-A-$}                        & $= \frac{3}{5}I_d - \frac{2}{5}A - \frac{2}{5}C$                 & 
\end{tabular}
\end{center}
with elements 
\begin{align*}
    Q_1 &= \frac{2}{5} A + \frac{1}{3}I_d,\\
    Q_2 &= -\frac{3}{10} A -\frac{1}{5} B -\frac{1}{5} C + \frac{2}{5}I_d,\\
    Q_3 &= -\frac{1}{10} A -\frac{1}{5} B +\frac{1}{5} C + \frac{1}{5}I_d,\\
    Q_4 &= -\frac{1}{10} A + \frac{1}{5} B - \frac{1}{5} C + \frac{1}{5}I_d,\\
    Q_5 &= \frac{1}{10} A + \frac{1}{5} B + \frac{1}{5} C.
\end{align*}

For the example in Eq.~\eqref{eq:example-123-145}, this means that after adding sufficient noise, the joint POVM has elements
\begin{center}
\begin{tabular}{cccl}
\cline{1-3}
\multicolumn{1}{|c|}{$Q_1$} & \multicolumn{1}{c|}{0}    & \multicolumn{1}{c|}{0}    & $=\frac{2}{5}I_2$                       \\ \cline{1-3}
\multicolumn{1}{|c|}{0}    & \multicolumn{1}{c|}{$Q_5$} & \multicolumn{1}{c|}{$Q_4$} & $=\frac{1}{5} e_1 e_1^\ast + \frac{1}{5}I_2$                       \\ \cline{1-3}
\multicolumn{1}{|c|}{0}    & \multicolumn{1}{c|}{$Q_3$} & \multicolumn{1}{c|}{$Q_2$} & $=\frac{1}{5} e_2 e_2^\ast + \frac{1}{5}I_2$                \\ \cline{1-3}
\phantom{$I_d$}$=\frac{2}{5}I_2$\phantom{$-A-$}                         & \phantom{$I_d$}$=\frac{1}{5} f_1 f_1^\ast + \frac{1}{5}I_2$\phantom{$-A-$}                        & $=\frac{1}{5} f_2 f_2^\ast + \frac{1}{5}I_2$                 & 
\end{tabular}
\end{center}
where 
\begin{align*}
    Q_1 &= \frac{2}{5}I_2,\\
    Q_2 &= \frac{1}{10}\left( \frac{1}{2} I_2 + f_2 f_2^\ast + e_2 e_2^\ast\right),\\
    Q_3 &= \frac{1}{10}\left( \frac{1}{2} I_2 + f_1 f_1^\ast + e_2 e_2^\ast\right),\\
    Q_4 &= \frac{1}{10}\left( \frac{1}{2} I_2 + f_2 f_2^\ast + e_1 e_1^\ast\right),\\
    Q_5 &= \frac{1}{10}\left( \frac{1}{2} I_2 + f_1 f_1^\ast + e_1 e_1^\ast\right).
\end{align*}
Again, we have written $e_1$, $e_2$ for the standard basis vectors in $\mathbb C^2$ and $f_1 = 1/\sqrt{2}(e_1 + e_2)$, $f_2 = 1/\sqrt{2}(e_1 - e_2)$. In this case, it is easy to see that the $Q_i$ indeed form a POVM. However, it can be checked with a semidefinite program inspired by \cite{wolf2009measurements} that we could have taken $s = 3/4> 2/5 $ in this case. Hence the question remains open if $s=2/5$ is optimal in this setup. Note that if $s$ is an inclusion constant for the polytope with square basis, $s \leq \frac{1}{\sqrt{2}}$ is a necessary requirement by the results in \cite{bluhm2018joint}, since for $A=0$, the problem reduces to the compatibility of two dichotomic POVMs.

\bigskip

\noindent\textbf{Acknowledgements.} 

The authors would like to thank Eric Evert for discussions concerning Lemma \ref{lem:arveson-EP-min}. Moreover, the authors would like to thank the anonymous referee for Lemma \ref{lem:equivalent-claim} and Theorem \ref{thm:A-almost-commuting}, hereby answering a question in a previous draft of this article. I.N.~was supported by the ANR projects \href{https://esquisses.math.cnrs.fr/}{ESQuisses}, grant number ANR-20-CE47-0014-01 and \href{https://www.math.univ-toulouse.fr/~gcebron/STARS.php}{STARS}, grant number ANR-20-CE40-0008, and by the PHC program \emph{Star} (Applications of random matrix theory and abstract harmonic analysis to quantum information theory). S.S. has received funding from the European Union's Horizon 2020 research and innovation programme under the Marie Sklodowska-Curie grant agreement No. 101030346. 

\bibliographystyle{unsrt}
\bibliography{spectralit}

\begin{thebibliography}{10}

\bibitem{bluhm2018joint}
Andreas Bluhm and Ion Nechita.
\newblock Joint measurability of quantum effects and the matrix diamond.
\newblock {\em Journal of Mathematical Physics}, 59(11):112202, 2018.

\bibitem{bluhm2020compatibility}
Andreas Bluhm and Ion Nechita.
\newblock Compatibility of quantum measurements and inclusion constants for the
  matrix jewel.
\newblock {\em SIAM Journal on Applied Algebra and Geometry}, 4(2):255--296,
  2020.

\bibitem{Bluhm2022norms}
Andreas Bluhm and Ion Nechita.
\newblock A tensor norm approach to quantum compatibility.
\newblock {\em Journal of Mathematical Physics}, 63:062201, 2022.

\bibitem{Heisenberg1927}
Werner Heisenberg.
\newblock {{\"U}ber den anschaulichen Inhalt der quantentheoretischen Kinematik
  und Mechanik}.
\newblock {\em Zeitschrift f{\"u}r Physik}, 43(3):172--198, 1927.

\bibitem{Bohr1928}
Niels Bohr.
\newblock The quantum postulate and the recent development of atomic theory.
\newblock {\em Nature}, 121(3050):580--590, 1928.

\bibitem{Fine1982}
Arthur Fine.
\newblock {Hidden variables, joint probability, and the Bell inequalities}.
\newblock {\em Physical Review Letters}, 48(5):291--295, 1982.

\bibitem{Brunner2014}
Nicolas {Brunner}, Daniel {Cavalcanti}, Stefano {Pironio}, Valerio {Scarani},
  and Stephanie {Wehner}.
\newblock {Bell nonlocality}.
\newblock {\em Reviews of Modern Physics}, 86:419--478, 2014.

\bibitem{Heinosaari2015}
Teiko Heinosaari, Jukka Kiukas, and Daniel Reitzner.
\newblock Noise robustness of the incompatibility of quantum measurements.
\newblock {\em Physical Review A}, 92:022115, 2015.

\bibitem{cuevas2020quantum}
Gemma De~las Cuevas, Tom Drescher, and Tim Netzer.
\newblock Quantum magic squares: Dilations and their limitations.
\newblock {\em Journal of Mathematical Physics}, 61(11):111704, 2020.

\bibitem{guerini2018joint}
Leonardo Guerini and Alexandre Baraviera.
\newblock Joint measurability meets {B}irkhoff-von {N}eumann's theorem.
\newblock {\em arXiv preprint arXiv:1809.07366v3}, 2019.

\bibitem{helton2019dilations}
J.~William Helton, Igor Klep, Scott McCullough, and Markus Schweighofer.
\newblock Dilations, linear matrix inequalities, the matrix cube problem and
  beta distributions.
\newblock {\em Memoirs of the American Mathematical Society}, 257(1232), 2019.

\bibitem{li2000convexity}
Chi-Kwong Li and Yiu-Tung Poon.
\newblock Convexity of the joint numerical range.
\newblock {\em SIAM Journal on Matrix Analysis and Applications},
  21(2):668--678, 2000.

\bibitem{gutkin2004convexity}
Eugene Gutkin, Edmond~A Jonckheere, and Michael Karow.
\newblock Convexity of the joint numerical range: topological and differential
  geometric viewpoints.
\newblock {\em Linear Algebra and its Applications}, 376:143--171, 2004.

\bibitem{Bremner1997}
David~D. Bremner.
\newblock On the complexity of vertex and facet enumeration for complex
  polytopes.
\newblock {\em Ph.D. thesis, School of Computer Science, McGill University,
  Monr{\'e}al, Canada}, 1997.

\bibitem{davidson2016dilations}
Kenneth~R. Davidson, Adam Dor-On, Orr~Moshe Shalit, and Baruch Solel.
\newblock Dilations, inclusions of matrix convex sets, and completely positive
  maps.
\newblock {\em International Mathematics Research Notices},
  2017(13):4069--4130, 2017.

\bibitem{passer2018minimal}
Benjamin Passer, Orr~Moshe Shalit, and Baruch Solel.
\newblock Minimal and maximal matrix convex sets.
\newblock {\em Journal of Functional Analysis}, 274:3197--3253, 2018.

\bibitem{Passer2019}
Benjamin Passer and Vern~I. Paulsen.
\newblock Matrix range characterizations of operator system properties.
\newblock {\em Journal of Operator Theory}, 85:547--568, 2021.

\bibitem{Heinosaari2011}
Teiko Heinosaari and M{\'a}rio Ziman.
\newblock {\em The Mathematical Language of Quantum Theory}.
\newblock Cambridge University Press, 2011.

\bibitem{watrous2018theory}
John Watrous.
\newblock {\em The Theory of Quantum Information}.
\newblock Cambridge University Press, 2018.

\bibitem{Heinosaari2016}
Teiko Heinosaari, Takayuki Miyadera, and M\'ario Ziman.
\newblock An invitation to quantum incompatibility.
\newblock {\em Journal of Physics A: Mathematical and Theoretical},
  49(12):123001, 2016.

\bibitem{lami2018non}
Ludovico Lami.
\newblock Non-classical correlations in quantum mechanics and beyond.
\newblock {\em PhD thesis. arXiv preprint arXiv:1803.02902}, 2018.

\bibitem{Barvinok2002}
Alexander Barvinok.
\newblock {\em A course in convexity}, volume~54 of {\em Graduate Studies in
  Mathematics}.
\newblock American Mathematical Society, 2002.

\bibitem{aliprantis2013infinite}
Charalambos~D. Aliprantis and Kim~C. Border.
\newblock {\em Infinite Dimensional Analysis: A Hitchhiker's Guide}.
\newblock Springer, 3rd edition, 2006.

\bibitem{bronstein2008approximation}
Efim~M. Bronstein.
\newblock Approximation of convex sets by polytopes.
\newblock {\em Journal of Mathematical Sciences}, 153(6):727--762, 2008.

\bibitem{evert2018extreme}
Eric Evert, J.~William Helton, Igor Klep, and Scott McCullough.
\newblock Extreme points of matrix convex sets, free spectrahedra, and dilation
  theory.
\newblock {\em The Journal of Geometric Analysis}, 28:1373--1408, 2018.

\bibitem{passer2022complex}
Benjamin Passer.
\newblock Complex free spectrahedra, absolute extreme points, and dilations.
\newblock {\em Documenta Mathematica}, 27:1275--1297, 2022.

\bibitem{li2020joint}
Chi-Kwong Li, Yiu-Tung Poon, and Ya-Shu Wang.
\newblock Joint numerical ranges and commutativity of matrices.
\newblock {\em Journal of Mathematical Analysis and Applications},
  491(1):124310, 2020.

\bibitem{fritz2017spectrahedral}
Tobias Fritz, Tim Netzer, and Andreas Thom.
\newblock Spectrahedral containment and operator systems with
  finite-dimensional realization.
\newblock {\em SIAM Journal on Applied Algebra and Geometry}, 1(1):556--574,
  2017.

\bibitem{huber2021note}
Beatrix Huber and Tim Netzer.
\newblock A note on non-commutative polytopes and polyhedra.
\newblock {\em Advances in Geometry}, 21(1):119--124, 2021.

\bibitem{aubrun2021entangleability}
Guillaume Aubrun, Ludovico Lami, Carlos Palazuelos, and Martin Pl{\'a}vala.
\newblock Entangleability of cones.
\newblock {\em Geometric and Functional Analysis}, 31:181--205, 2021.

\bibitem{jencova2018incompatible}
Anna Jen{\v{c}}ov{\'a}.
\newblock Incompatible measurements in a class of general probabilistic
  theories.
\newblock {\em Physical Review A}, 98(1):012133, 2018.

\bibitem{Bluhm2022GPT}
Andreas Bluhm, Anna Jen{\v{c}}ov{\'a}, and Ion Nechita.
\newblock Incompatibility in general probabilistic theories, generalized
  spectrahedra, and tensor norms.
\newblock {\em Communications in Mathematical Physics}, 393(3):1125--1198,
  2022.

\bibitem{banica2007quantum}
Teodor Banica, Julien Bichon, and Beno{\^\i}t Collins.
\newblock Quantum permutation groups: a survey.
\newblock {\em Banach Center Publications}, 78(1):13--34, 2007.

\bibitem{ziegler2012lectures}
G{\"u}nter~M Ziegler.
\newblock {\em Lectures on polytopes}, volume 152.
\newblock Springer Science \& Business Media, 2012.

\bibitem{brualdi1977convex}
Richard~A. Brualdi and Peter~M. Gibson.
\newblock {Convex polyhedra of doubly stochastic matrices. I. Applications of
  the permanent function}.
\newblock {\em Journal of Combinatorial Theory, Series A}, 22(2):194--230,
  1977.

\bibitem{acin2015combinatorial}
Antonio Ac{\'\i}n, Tobias Fritz, Anthony Leverrier, and Ana~Bel{\'e}n Sainz.
\newblock A combinatorial approach to nonlocality and contextuality.
\newblock {\em Communications in Mathematical Physics}, 334:533--628, 2015.

\bibitem{shultz1974characterization}
Frederic~W.\ Shultz.
\newblock A characterization of state spaces of orthomodular lattices.
\newblock {\em Journal of Combinatorial Theory, Series A}, 17(3):317--328,
  1974.

\bibitem{bluhm2022maximal}
Andreas Bluhm and Ion Nechita.
\newblock Maximal violation of steering inequalities and the matrix cube.
\newblock {\em Quantum}, 6:656, 2022.

\bibitem{van2001course}
Jacobus~Hendricus Van~Lint and Richard~Michael Wilson.
\newblock {\em A course in combinatorics}.
\newblock Cambridge University Press, second edition, 2001.

\bibitem{yannakakis1991expressing}
Mihalis Yannakakis.
\newblock Expressing combinatorial optimization problems by linear programs.
\newblock {\em Journal of Computer and System Sciences}, 43(3):441--466, 1991.

\bibitem{wolf2009measurements}
Michael~M. Wolf, David P\'erez-Garc\'ia, and Carlos Fern\'andez.
\newblock Measurements incompatible in quantum theory cannot be measured
  jointly in any other no-signaling theory.
\newblock {\em Physical Review Letters}, 103:230402, 2009.

\end{thebibliography}

\end{document}